\documentclass[a4paper,UKenglish,cleveref, autoref, thm-restate]{lipics-v2021}

\pdfoutput=1 
\hideLIPIcs  

\title{\texorpdfstring{FPT Approximations for Fair $k$-Min-Sum-Radii}{FPT Approximations for Fair k-Min-Sum-Radii}} 

\author{Lena Carta}{University of Bonn, Germany}{}{}{}

 \author{Lukas Drexler}{Heinrich Heine University D\"usseldorf, Germany}{lukas.drexler@hhu.de}{https://orcid.org/0000-0001-9395-6711}{}

 \author{Annika Hennes\footnote{Corresponding author}}{Heinrich Heine University D\"usseldorf, Germany}{annika.hennes@hhu.de}{https://orcid.org/0000-0001-9109-3107}{}

 \author{Clemens R\"osner}{Fraunhofer-Institut für Algorithmen und Wissenschaftliches Rechnen SCAI, Germany}{clemens.roesner@gmx.de}{}{}

 \author{Melanie Schmidt}{Heinrich Heine University D\"usseldorf, Germany}{mschmidt@hhu.de}{https://orcid.org/0000-0003-4856-3905}{}

\authorrunning{L. Carta, L. Drexler, A. Hennes, C. R\"osner, and M. Schmidt} 

\Copyright{Lena Carta, Lukas Drexler, Annika Hennes, Clemens R\"osner, and Melanie Schmidt} 



\category{} 

\relatedversion{} 

\funding{Deutsche Forschungsgemeinschaft (DFG, German Research Foundation) - Project 456558332.}

\nolinenumbers 

\EventEditors{Juli\'{a}n Mestre and Anthony Wirth}
\EventNoEds{2}
\EventLongTitle{35th International Symposium on Algorithms and Computation (ISAAC 2024)}
\EventShortTitle{ISAAC 2024}
\EventAcronym{ISAAC}
\EventYear{2024}
\EventDate{December 8--11, 2024}
\EventLocation{Sydney, Australia}
\EventLogo{}
\SeriesVolume{322}
\ArticleNo{53}

\usepackage{graphicx} 
\usepackage{amsmath,amsfonts,amsthm, amssymb}
\usepackage{mathtools}
\usepackage{todonotes}
\usepackage{xspace}
\usepackage[ruled,nofillcomment, ruled, vlined, linesnumbered]{algorithm2e}

\SetCommentSty{mycommfont}
\SetAlgoCaptionLayout{textsc} 
\SetProcNameSty{textsc}
\usepackage{thmtools}
\usetikzlibrary{calc}

\usepackage{hyperref}
\usepackage[hyphenbreaks]{breakurl}

\DeclareMathOperator{\poly}{poly}
\DeclareMathOperator{\MSR}{MSR}
\DeclareMathOperator{\OPT}{OPT}

\DeclareMathOperator{\eps}{\varepsilon}

\DeclareMathOperator{\CC}{CC}
\DeclareMathOperator{\V}{V}

\newcommand{\rCC}[2]{\hat{r}^{#1}_{#2}} 
\newcommand{\cCC}[2]{\hat{c}^{#1}_{#2}} 

\DeclarePairedDelimiter\ceil{\lceil}{\rceil}

\usepackage{booktabs}
\usepackage{footnote}
\makesavenoteenv{tabular}

\begin{document}

\maketitle

\begin{abstract}
    We consider the $k$-min-sum-radii ($k$-MSR) clustering problem with fairness constraints. The $k$-min-sum-radii problem is a mixture of the classical $k$-center and $k$-median problems. We are given a set of points $P$ in a metric space and a number $k$ and aim to partition the points into $k$ clusters, each of the clusters having one designated center. The objective to minimize is the sum of the radii of the $k$ clusters (where in $k$-center we would only consider the maximum radius and in $k$-median we would consider the sum of the individual points' costs).

    Various notions of fair clustering have been introduced lately, and we follow the definitions due to Chierichetti et al.~\cite{CKLV17} which demand that cluster compositions shall follow the proportions of the input point set with respect to some given sensitive attribute. For the easier case where the sensitive attribute only has two possible values and each is equally frequent in the input, the aim is to compute a clustering where all clusters have a 1:1 ratio with respect to this attribute. We call this the 1:1 case.

    There has been a surge of FPT-approximation algorithms for the $k$-MSR problem lately, solving the problem both in the unconstrained case and in several constrained problem variants. We add to this research area by designing an FPT $(6+\epsilon)$-approximation that works for $k$-MSR under the mentioned general fairness notion.  For the special 1:1 case, we improve our algorithm to achieve a $(3+\epsilon)$-approximation.
\end{abstract}
\clearpage
\setcounter{page}{1}

\section{Introduction}

Cluster analysis is an unsupervised learning task that has inspired much research during the last decades. Nearly all popular clustering formulations lead to NP-hard and often APX-hard problems, thus there is a thriving field designing approximation algorithms for clustering. A very popular and well studied problem is the \emph{$k$-median} problem: Given $n$ points $P$ from a metric space and a number $k$, find $k$ centers $C\subseteq P$ such that the sum of the distances of all points to their respective centers is minimized. Notice that $k$ centers implicitly define $k$ clusters by assigning each point to its closest center (breaking ties arbitrarily). 
The $k$-median problem is APX-hard~\cite{GK99,JMS02}, but allows for $O(1)$-approximations. After a long line of research, the currently best approximation for $k$-median achieves a guarantee of~$2.675+\varepsilon$~\cite{BPRST15}\footnote{This result actually holds for a slightly more general variant where the set of input points $P$ can be different than the set of possible center locations from which we pick $C$.}. A clustering function that is very popular for its simplicity and elegant algorithms is \emph{$k$-center}. It has the same input and solution space, but judges clusterings based on the maximum radius of the (induced) clusters. The goal is to minimize the radius of the cluster with largest radius. It has been known for quite some time that $k$-center admits a $2$-approximation~\cite{G85,HS85}, and that this is tight when assuming P$\neq$ NP~\cite{HN79}.

In this paper, we consider \emph{$k$-min-sum-radii} clustering, abbreviated as $k$-MSR. We get the same input and solution space as for $k$-center, but the objective changes to the \emph{sum} of the cluster radii. The problem thus lies in between $k$-center and $k$-median as it is a sum-based objective, but considers radii instead of points. 
Contrary to intuition, in metric spaces, many design techniques fail for $k$-msr which work fine for $k$-center and $k$-median. However, the problem allows for a polynomial $(3+\epsilon)$-approximation~\cite{buchem2024} via primal dual algorithms. 

The model of fair clustering was introduced by Chierichetti et. al.~\cite{CKLV17} based on a disparate impact approach. In the easiest case, given an input point set with the same number of blue and red points, the goal of fair clustering is to find a clustering where every cluster is composed of the same number of red and blue points (the colors represent values of a sensitive attribute). The number of points in a cluster is unlimited, but the composition of the clusters is constrained. For $k$-center \emph{and} $k$-median, the following simple approach suffices to design polynomial-time approximation algorithms for fair clustering in this scenario: 

Compute a fair micro clustering, i.e., a clustering into $\mu \gg k$ clusters which is fair (this is easier than finding exactly $k$ clusters, for two colors and $\mu=n/2$ it is just a matching). Consider the clusters in this micro clustering as inseparable entities and use an algorithm for the unconstrained problem to cluster them into $k$ clusters. The resulting clusters are fair because the union of fair clusters is again fair (this property of the constraint \emph{fairness} is sometimes called \emph{mergeability}). In addition, both for $k$-center and for $k$-median, the cost of the resulting clusters can be reasonably bounded, yielding $O(1)$-approximations for the respective fair clustering problems. 
Most surprisingly, the same approach does \emph{not} work in the case of $k$-min-sum-radii. The reason is that for $k$-min-sum-radii, the cost of a micro clustering can actually be \emph{larger} than the clustering with $k$ clusters. This property is shared neither by $k$-center nor by $k$-median: For $k$-center, the radius of the micro clustering is always bounded by the $k$-clustering, and for $k$-median, the sum of the points' costs in the micro clustering is bounded by the respective sum of the $k$-clustering.

\begin{figure}
\begin{tikzpicture}
\centering
\begin{scope}
    \node (a) [circle,draw,fill,inner sep=0cm,minimum size=0.1cm] at (0,0)  {};
    \node (b) [circle,draw,fill,inner sep=0cm,minimum size=0.1cm] at (0.5,0)  {};
    \node (c) [circle,draw,fill,inner sep=0cm,minimum size=0.1cm] at (1,0)  {};
    \node (d) [circle,draw,fill,inner sep=0cm,minimum size=0.1cm] at (1.5,0)  {};
    \draw [thick] (a) -- (b) -- (c) to (d);
    \draw [thick, bend left] (a) to (c);
    \draw [thick, bend left] (b) to (d);
    \draw [thick, bend left] (a) to (d);
    \node (e) [circle,draw,fill,inner sep=0cm,minimum size=0.1cm] at (0,-0.5)  {};
    \node (f) [circle,draw,fill,inner sep=0cm,minimum size=0.1cm] at (0.5,-0.5)  {};
    \node (g) [circle,draw,fill,inner sep=0cm,minimum size=0.1cm] at (1,-0.5)  {};
    \node (h) [circle,draw,fill,inner sep=0cm,minimum size=0.1cm] at (1.5,-0.5)  {};
    \draw [thick] (e) -- (f) -- (g) to (h);
    \draw [thick, bend right] (e) to (g);
    \draw [thick, bend right] (f) to (h);
    \draw [thick, bend right] (e) to (h);
    \draw [thick] (a) -- (e);
    \draw [thick] (b) -- (f);
    \draw [thick] (c) -- (g);
    \draw [thick] (d) -- (h);
    \fill [opacity=0.2] (0.75,-0.25) ellipse (1.2 and 0.6);
    \node at (0.75,-1) {cost $1$};
\end{scope}
\node at (2.25,-0.25) {$\Rightarrow$};
\begin{scope}[xshift=2.75cm]
    \node (a) [circle,draw,fill,inner sep=0cm,minimum size=0.1cm] at (0,0)  {};
    \node (b) [circle,draw,fill,inner sep=0cm,minimum size=0.1cm] at (0.5,0)  {};
    \node (c) [circle,draw,fill,inner sep=0cm,minimum size=0.1cm] at (1,0)  {};
    \node (d) [circle,draw,fill,inner sep=0cm,minimum size=0.1cm] at (1.5,0)  {};
    \draw [thick] (a) -- (b) -- (c) to (d);
    \draw [thick, bend left] (a) to (c);
    \draw [thick, bend left] (b) to (d);
    \draw [thick, bend left] (a) to (d);
    \node (e) [circle,draw,fill,inner sep=0cm,minimum size=0.1cm] at (0,-0.5)  {};
    \node (f) [circle,draw,fill,inner sep=0cm,minimum size=0.1cm] at (0.5,-0.5)  {};
    \node (g) [circle,draw,fill,inner sep=0cm,minimum size=0.1cm] at (1,-0.5)  {};
    \node (h) [circle,draw,fill,inner sep=0cm,minimum size=0.1cm] at (1.5,-0.5)  {};
    \draw [thick] (e) -- (f) -- (g) to (h);
    \draw [thick, bend right] (e) to (g);
    \draw [thick, bend right] (f) to (h);
    \draw [thick, bend right] (e) to (h);
    \draw [thick] (a) -- (e);
    \draw [thick] (b) -- (f);
    \draw [thick] (c) -- (g);
    \draw [thick] (d) -- (h);
    \fill [opacity=0.2] (0,-0.25) ellipse (0.2 and 0.6);
    \fill [opacity=0.2] (0.5,-0.25) ellipse (0.2 and 0.6);
    \fill [opacity=0.2] (1,-0.25) ellipse (0.2 and 0.6);
    \fill [opacity=0.2] (1.5,-0.25) ellipse (0.2 and 0.6);
    \node at (0.75,-1) {cost $\mu$};
\end{scope}
\begin{scope}[xshift=5.25cm]
    \node (a) [circle,draw,fill,inner sep=0cm,minimum size=0.1cm] at (0.5,0)  {};
    \node (e) [circle,draw,fill,inner sep=0cm,minimum size=0.1cm] at (0,0)  {};
    \node (f) [circle,draw,fill,inner sep=0cm,minimum size=0.1cm] at (0.25,-0.5)  {};
    \node (g) [circle,draw,fill,inner sep=0cm,minimum size=0.1cm] at (0.75,-0.5)  {};
    \node (h) [circle,draw,fill,inner sep=0cm,minimum size=0.1cm] at (1,0)  {};
    \draw [thick] (a) -- (e);
    \draw [thick] (a) -- (f);
    \draw [thick] (a) -- (g);
    \draw [thick] (a) -- (h);
    \fill [opacity=0.2] (0.5,-0.25) ellipse (0.8 and 0.6);
    \node at (0.5,-1) {cost $1$};
\end{scope}
\node at (6.75,-0.25) {$\Rightarrow$};
\begin{scope}[xshift=7.25cm]
    \node (a) [circle,draw,fill,inner sep=0cm,minimum size=0.1cm] at (0.5,0)  {};
    \node (e) [circle,draw,fill,inner sep=0cm,minimum size=0.1cm] at (0,0)  {};
    \node (f) [circle,draw,fill,inner sep=0cm,minimum size=0.1cm] at (0.25,-0.5)  {};
    \node (g) [circle,draw,fill,inner sep=0cm,minimum size=0.1cm] at (0.75,-0.5)  {};
    \node (h) [circle,draw,fill,inner sep=0cm,minimum size=0.1cm] at (1,0)  {};    \draw [thick] (a) -- (e);
    \draw [thick] (a) -- (f);
    \draw [thick] (a) -- (g);
    \draw [thick] (a) -- (h);
    \fill [opacity=0.2] (0.25,0) ellipse (0.5 and 0.2);
    \fill [opacity=0.2] (f) ellipse (0.15 and 0.15);
    \fill [opacity=0.2] (g) ellipse (0.15 and 0.15);
    \fill [opacity=0.2] (h) ellipse (0.15 and 0.15);
    \node at (0.5,-1) {cost $1$};
\end{scope}
\begin{scope}[xshift=9.5cm]
    \node (a) [circle,draw,fill,inner sep=0cm,minimum size=0.1cm] at (0,0)  {};
    \node (b) [circle,draw,fill,inner sep=0cm,minimum size=0.1cm] at (0.5,0)  {};
    \node (c) [circle,draw,fill,inner sep=0cm,minimum size=0.1cm] at (0.5,-0.5)  {};
    \node (d) [circle,draw,fill,inner sep=0cm,minimum size=0.1cm] at (0,-0.5)  {};
    \draw [thick] (a) -- (b) -- (c) -- (d) -- (a);
    \draw [thick] (a) -- (c);
    \fill [opacity=0.2] (0.25,-0.25) ellipse (0.55 and 0.55);
    \node at (0.25,-1) {cost $1$};
\end{scope}
\node at (10.75,-0.25) {$\Rightarrow$};
\begin{scope}[xshift=11.5cm]
    \node (a) [circle,draw,fill,inner sep=0cm,minimum size=0.1cm] at (0,0)  {};
    \node (b) [circle,draw,fill,inner sep=0cm,minimum size=0.1cm] at (0.5,0)  {};
    \node (c) [circle,draw,fill,inner sep=0cm,minimum size=0.1cm] at (0.5,-0.5)  {};
    \node (d) [circle,draw,fill,inner sep=0cm,minimum size=0.1cm] at (0,-0.5)  {};
    \draw [thick] (a) -- (b) -- (c) -- (d) -- (a);
    \draw [thick] (a) -- (c);
    \fill [opacity=0.2] (a) ellipse (0.15 and 0.15);
    \fill [opacity=0.2] (b) ellipse (0.15 and 0.15);
    \fill [opacity=0.2] (c) ellipse (0.15 and 0.15);
    \fill [opacity=0.2] (d) ellipse (0.15 and 0.15);
    \node at (0.25,-1) {cost $0$};
\end{scope}
\node at (2.25,-1.5) {(a) cost increases to $\mu$};
\node at (6.75,-1.5) {(b) cost stays the same};
\node at (10.75,-1.5) {(c) cost drops to $0$};

\end{tikzpicture}
\caption{Anything can happen for $k$-MSR: The cost of a micro clustering with $\mu=4$ compared to the macro clustering with $k=1$. All pictures use shortest path metrics, the edges have unit weights.\label{anything-goes}}
\end{figure}
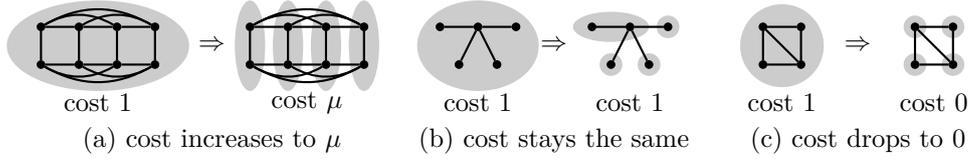
Figure~\ref{anything-goes} illustrates the situation for $k$-MSR. The figure is for unconstrained $k$-MSR to ease visualization, and just illustrates how the cost of micro clusterings for $k$-MSR can behave. The examples are depicted with $k=1$ and $\mu=4$. In (a), we see a cluster with $k$-MSR cost $1$ in which all points have the same pairwise distance, namely $1$. If this cluster is broken into $\mu$ pieces, then these pieces suddenly contribute $t$ to the objective. In (b), we have a star with $\mu$ leaves where one cost $1$ cluster remains when we use $\mu$ clusters, so the cost stays the same. 
In (c), the cluster only contains $\mu$ points, and then the cost drops to zero when it is divided into $\mu$ subclusters. Overall we observe that we have no control over the cost of a micro clustering and that the micro clustering approach fails. 

\subparagraph*{FPT Approximation}

\begin{table}[ht]
\centering
\caption{A list of FPT approximation results for $k$-MSR and $k$-median, all of which appeared in the last five years. Multiple results for the same problem are listed in reverse chronological order. }
\label{tab:fptapprox}
\begin{threeparttable}
\begin{tabular}{@{}lcc@{\hskip.18cm}ccc@{}}
\toprule
                 & Unconstr. & Capacities / Uniform Cap. & Matroid Con.& Fair Centers  \\ 
 \midrule
 $k$-MSR & $2\!+\!\epsilon$~\cite{chen2024parameterized} & $\approx 7.6$~\cite{JKY24}, $15\!+\!\epsilon$~\cite{BL023a}  /  $3$~\cite{JKY24}, $4\!+\!\epsilon$~\cite{BL023a}, $28$~\cite{IV20} & $9\!+\!\epsilon$~\cite{IV20} & $3\!+\!\epsilon$~\cite{chen2024parameterized}  \\
$k$-med & $1.546$~\cite{AL24}\tnote{$\ast$} & $3\!+\!\epsilon$~\cite{CAL19}, $7\!+\!\epsilon$~\cite{ABMM019} / no improvement & $2\!+\!\epsilon$~\cite{CG0LL19} & -
\end{tabular}
\begin{tablenotes}
\item[$\ast$] This result holds for the problem as described in the introduction. If centers can be chosen from a set possibly different from $P$ then the best known FPT approximation bound is $\approx 1.735\!+\!\epsilon$~\cite{CG0LL19}.
\end{tablenotes}
\end{threeparttable}
\end{table}

There is mainly one approach for designing polynomial-time approximation algorithms for $k$-MSR, which is a primal-dual approach that yielded the currently best known bounds for unconstrained $k$-MSR and $k$-MSR with lower bounds on the cluster sizes, or outliers~\cite{buchem2024}. This lack of diversity in the techniques to obtain approximation algorithms for $k$-MSR has lately led to a surge of FPT approximation algorithms for $k$-MSR, that was also inspired by a similar strong interest in FPT approximation algorithms for the $k$-median problem. FPT approximation algorithms have been obtained for various problem variants, see Table~\ref{tab:fptapprox}. The ``fair centers'' variant demands that the set of centers is fair rather than the clusters. For more details, see the paragraph about related work.
While the obtained approximation algorithms only work for small $k$, they are still of high interest due to the problem insights that they provide and also due to the fact that clustering with a small number of clusters is an important domain.

\paragraph*{Our Results}
We get the following results in FPT-time where the parameter is the number of clusters $k$.

\begin{itemize}
  \item A $(3+\epsilon)$-approximation for the fair clustering $k$-MSR problem when there are only two colors, both have exactly a ratio of 1:1 in the input point set, and we also want to achieve that exact ratio. We are not aware of any results on fair $k$-MSR in general metrics. 
  \item A $(6+\epsilon)$-approximation that works for a variety of more general fairness constraints, including all notions defined in~\cite{bercea2019cost,CKLV17,rosner2018privacy}. To the best of our knowledge, there are no previous results for this problem.
\end{itemize}
We also extend our approach for clustering with uniform lower bounds, and generally to the class of \emph{mergeable} constraints (see \Cref{def:mergeable-constraint}). Uniform lower bounds have been studied in the clustering literature to model anonymity~\cite{APFTKKZ10}: The constraint demands that every cluster contains a minimum number of $\ell$ points, i.e., that $|a(c)|\ge \ell$ for all $c \in C$. 
A polynomial-time $(3.5+\epsilon)$-approximation algorithm for $k$-MSR with lower bounds  is known~\cite{buchem2024}, and our FPT approximation algorithm achieves a $(3+\epsilon)$-guarantee.

\subparagraph*{Main Technical Contribution}
Our main technical contribution is the development of a completely novel approach to design FPT approximation algorithms for $k$-MSR. We give more details on this approach in the next paragraph, but the main gain from our approach is that we make it possible to use $k$-center algorithms as a subroutine via a clever branching  that we have not seen like this in the literature before. We believe this technique to be of independent interest for the design of FPT approximation algorithms for $k$-MSR.

Following this novel design scheme, it is possible to obtain $O(1)$-approximations for fair $k$-MSR and mergeable constraints in general. An additional technical contribution lies in reducing the factors to small constants. In particular in the general fairness case, obtaining the factor $6+\epsilon$ requires a clever bounding technique.

\subparagraph*{More Insight into the Scheme}
We first discuss another approach that does \emph{not} work for $k$-MSR. There is a fairly general idea to obtain FPT approximations for constrained $k$-clustering problems which we can think of in the following way: Compute a solution to the unconstrained clustering problem with an approximation algorithm (here, one can even use a bicriteria approximation which computes $\mathcal{O}(k)$ centers). This gives a set of centers $S$. Then \lq move\rq\ all points to their closest center, i.e., create an instance $I$ where all points lie at the $k$ centers in $S$. The optimum cost for the constrained problem on $I$ can be related to the optimum cost for the constrained problem on the original instance. Then solve $I$ with an algorithm that uses that the points all lie on $k$ locations (notice that the constraint will prevent us from simply opening one center at each location). The resulting solution is then translated back to the original instance. 

Adamczyk et al.~\cite{ABMM019} use this approach to develop an FPT $(7+\varepsilon)$-approximation for capacitated $k$-median. 
But here is another problem of the $k$-MSR cost function: Moving all points to centers of an approximate solution also has uncontrollable effects on the cost function. There seems to be no easy fix to this, and Inamdar and  Varadarajan in the introduction of~\cite{IV20} also notice that the approach does not seem to extend to $k$-MSR, so there is a need for new techniques to design FPT algorithms for $k$-MSR in general.

The starting idea of our approach is to use an algorithm for the unconstrained \emph{$k$-center} problem at the core of the algorithm (rather than for an unconstrained $k$-MSR problem, which would follow the above scheme). This means that we start off with a small approximation ratio of $2$. 
Notice that there is a connection between the value $F^*$ of an optimal $k$-center solution and the largest radius $r_{\text{max}}$ in an optimal $k$-min-sum-radii solution:  $r_{\text{max}}$ must lie in the interval $[F^*, kF^*]$. This  relation is obviously not tight enough to directly lead to an algorithm, but it can be used to find a near optimal  approximation $\hat{r}_{\max}$ for the largest radius (a proof can be found in the full version).
Now our first idea is that we can find the largest cluster in an $k$-MSR solution by running a $k$-center algorithm with $\hat{r}_{\max}$ and then guessing in which of the $k$-center clusters the largest $k$-MSR cluster lies. But this only works for the first cluster. To recurse, we have to eliminate this cluster from the input such that a $k$-center algorithm can find the next cluster. The main hurdle here is that we cannot simply delete the cluster because we might destroy the optimal fair assignment (which we do not know and cannot guess at this point). We resolve this problem by keeping all points but adjusting the metric to a (non-metric) distance function. Then we show that we can still solve the resulting problem by approximately solving a so called \emph{$k$-center completion problem}.

This problem might be of independent interest and could play a role in further $k$-MSR research: Basically, we hand the problem an incomplete set of centers $C'=\{c_1,\ldots,c_{\ell}\}$ with $\ell \le k$ together with radii $r_1,\ldots,r_{\ell}$ and ask it to find a $k$-center solution where points can be assigned to a center $c_i$ from $C'$ while paying $r_i$ less than the actual distance. We observe that the $k$-center completion problem can be solved by Gonzalez' $2$-approximation technique for $k$-center~\cite{G85}.

Using this modeling we are able to recursively find largest clusters. There is another technical problem, though. In every step, we are guessing the cluster in the completed solution returned from the $k$-center completion problem in which the optimal center of the next $k$-MSR cluster lies. But this may be one of the first $\ell$ clusters. Now the final crucial idea is that in this case, we opt to increase the radius of that previous cluster instead of opening a new cluster. It makes sense that for the $k$-MSR objective, this is a good idea: Often we can reduce cost by using less clusters instead of using overlapping clusters.

Combining these three ideas ((i) guessing clusters based on a $k$-center solution, (ii) modeling the elimination of clusters by using a completion problem, (iii) allowing clusters to grow), we obtain our main algorithm. 

Applying our scheme yields a covering of the input points where we know that the balls are reasonably small compared to the balls in an optimal solution, and every optimum ball is covered by one of our balls. However, we also need to compute the actual clustering. Doing so requires to resolve the cluster membership of points that are in multiple balls, while respecting fairness constraints. For the general fairness case, we do this by modelling the overlap of balls by a graph and computing connected components in this graph. We can then show that the radius of the components is not too large (see Section~\ref{subsec:assighment}).

\subparagraph*{Related Work}\label{sec:related-work}
Research on FPT algorithms for constrained and unconstrained $k$-MSR is highly active at the moment, see Table~\ref{tab:fptapprox}. Notice that the paper  by Chen et al.~\cite{chen2024parameterized} gives an algorithm for a problem called \emph{fair sum of radii}, but it refers to setting different from ours: While points also have colors, the fairness constraints are not imposed on the clusters but rather on the centers. More precisely, each color $i$ has its own associated value $k_i$ and any feasible solution has to use exactly $k_i$ centers from that color. We call this \emph{fair centers} in the table. All the results in the table are in general metric space, which is our setting. 

In the more restricted Euclidean case, Drexler et al.~\cite{DHLSW23} gave a PTAS for the fair $k$-MSR problem for constant $k$, however, the PTAS was faulty and a corrected version only exists for $k$-MSR with outliers, not for fair $k$-MSR. It is thus open to give a better result for Euclidean fair $k$-MSR. Bandyapadhyay et al.~\cite{bandyapadhyay2023fpt} give a $(2+\eps)$-approximation for the capacitated $k$-MSR problem whose runtime linearly depends on the dimension, and a $(1+\eps)$-approximation which depends exponentially on the dimension.

There are some results on poly-time approximation algorithms for $k$-MSR in general metrics, all following the primal-dual scheme. For the unconstrained case, the first was due to Charikar and Panigrahy~\cite{CP01}, and there are two recent improvements by Friggstad and Jamshidian~\cite{FJ22} and Buchem et al~\cite{buchem2024} (see below). 
The $k$-MSR problem with uniform lower bounds has been studied by Ahmadian and Swamy~\cite{AS16} who give a polynomial-time 3.83-approximation, and additionally give a $12.365$-approximation if outliers are additionally considered.
In~\cite{buchem2024}, Buchem et al.\ improve upon all these factors by proposing a $(3+\eps)$-approximation for the unconstrained case and the version with outliers, and a $(3.5+\eps)$-approximation for lower bounds \emph{and} lower bounds with outliers. 
Their algorithm also works for the non-uniform lower bounded case (for this case, we could achieve a $(6+\epsilon$)-approximation since lower bounds are mergeable, but that is worse than $(3.5+\epsilon)$).

\newcommand{\centerc}{\texttt{Center}}
\newcommand{\msr}{\texttt{MSR}}

\section{Getting Started} \label{sec:prelims}
\subparagraph*{\texorpdfstring{Defining the $k$-Center and $k$-MSR Problem}{Defining the k-Center and k-MSR Problem}}
    An instance $I$ for a $k$-clustering problem consists of a finite set $P$ of $n$ points, a (metric) distance function $d:P\times P \to \mathbb{R}_{\geq 0}$ and a number $k\in \mathbb{N}$ with $k \leq n$. 
    A feasible solution $S = (\mathscr{C},\sigma)$ consists of a set $\mathscr{C} = \{c_1, \ldots, c_k\}\subseteq P$ of \emph{centers} and an assignment $\sigma: P \to \mathscr{C}$ of points to centers.
    For a center $c_i\in \mathscr{C}$, we call $C_i = \sigma^{-1}(c_i)$ the \emph{cluster} of $c_i$ induced by $\sigma$. Furthermore, we let $r_i = \max_{p\in C_i} d(p, c_i)$ denote the \emph{radius} of $C_i$.
    Let $r_1, r_2, \ldots, r_k$ be the radii induced by a solution $S$. We refer to the tuple $(r_1, \ldots, r_k)$ as the \emph{radius profile} of $S$.
    The cost of $S = (\mathscr{C},\sigma)$ with radius profile $(r_1, \ldots, r_k)$ with respect to the $k$-center objective is defined as 
    $
    \max_{i\in\{1,\ldots, k\}} r_i,
    $    
    while the cost with respect to the $k$-min-sum-radii objective is defined as 
$
    \msr(S) = \sum_{i=1}^{k} r_i.
  $ 
    The $k$-center/$k$-min-sum-radii problem takes as input a point set $P$, a metric $d \colon P \times P \to \mathbb{R}_{\ge 0}$ and a number $k \in \mathbb{N}$, and the task is to minimize the $k$-center/$k$-min-sum-radii objective.

\subparagraph*{Exact Fairness}

In the fairness setting, every point belongs to exactly one protected group. Here, we will usually denote these groups by \textit{colors} $\gamma_1,\ldots,\gamma_m$. A coloring of the points is given by a function $\gamma\colon P \to \{\gamma_1,\ldots,\gamma_m\}$.

The notion of \textit{exact fairness} as for example defined in \cite{rosner2018privacy} is based on maintaining the underlying proportions of colors in the clusters. That is, for every color $\gamma_j$, the proportion of points in $P$ with color $\gamma_j$ is the same as the proportion of points with color $\gamma_j$ in any cluster. To be more precise, we call a solution $S = (\mathscr{C}, \sigma)$ with induced clusters $C_1,\ldots,C_k$ \textit{fair} if
\[ \frac{\lvert C_i \cap \Gamma_j \rvert}{\lvert C_i \rvert} = \frac{\lvert\Gamma_j\rvert}{\lvert P \rvert} \]
for all $i\leq k$ and $j\leq m$, where $\Gamma_j = \gamma^{-1}(\gamma_j)$ is the set of points with color $\gamma_j$.
The special case of $m=2$ and $|\Gamma_1|=|\Gamma_2|$ has a specifically nice structure because the optimum solution can be partitioned into $|P|/2$ bicolored pairs. We refer to this as the 1:1 case.

There also exist more relaxed definitions of fairness that do not demand strict preservation of input ratios. In the full version, we discuss notions from \cite{bera2019fair,bercea2019cost, bohm2020fair,CKLV17,rosner2018privacy}.

\subparagraph*{Mergeable Constraints}
Let $(\mathscr{C}, \sigma)$ be a clustering. We \emph{merge} two clusters $C_i\coloneqq\sigma^{-1}(c_i)$, $C_j\coloneqq\sigma^{-1}(c_j)$ by replacing $c_i, c_j \in \mathscr{C}$ by an arbitrary point $c'\in C_i \cup C_j$ (i.e., $\mathscr{C} \coloneqq \mathscr{C} \setminus \{c_i, c_j\} \cup \{c'\}$) and assigning $\sigma(p) = c'$ for all $p\in C_i \cup C_j$.
Notice that because $\sigma$ maps every point to exactly one center, all clusters $C_i \coloneqq \sigma^{-1}(c_i),\ i=1,\ldots,k$ need to be pairwise disjoint.

\begin{definition}\label{def:mergeable-constraint}
    We say a constraint is \emph{mergeable} if a feasible clustering is still feasible with respect to the constraint after merging two of its clusters.
\end{definition}
In this context, we say an assignment is \textit{feasible} if it preserves the constraint.
When dealing with the $k$-min-sum-radii problem with a specific mergeable constraint, we refer to the $k$-center problem with the same constraint as the \emph{corresponding} $k$-center problem. For example, for $k$-min-sum-radii with exact fairness, the corresponding $k$-center problem is $k$-center with exact fairness.
In the full version, we list several (fairness) constraints that are mergeable and prove this.

\subsection{\texorpdfstring{The $k$-center completion problem}{The k-Center Completion Problem}}

\newcommand{\kcenterproblem}{$k$-center completion problem\xspace}
\newcommand{\kminsumradiiproblem}{$k$-min-sum-radii completion problem\xspace}
\newcommand{\kcentersolution}{$k$-center completion solution\xspace}
\newcommand{\kminsumradiisolution}{$k$-min-sum-radii completion solution\xspace}

\newcommand{\centercompletion}{\texttt{Center-Completion}}

The following problem can be solved in a relatively straightforward way using Gonzalez' algorithm~\cite{G85}. 

\begin{definition}
    The \emph{\kcenterproblem}  takes as input a point set $P$, a metric $d \colon P \times P \to \mathbb{R}_{\ge 0}$, a number $k \in \mathbb{N}$, a set of predefined centers $c_1, \ldots, c_{\ell}$ for an $\ell \in [k]$, and corresponding radii $r_1,\ldots,r_{\ell}$. The aim is to  compute centers $c_{\ell+1},\ldots,c_k$ and an assignment $\alpha\colon P \to \{c_1,\ldots,c_k\}$ such that 
    $
    \max_{x \in P} d'(x,\alpha(x)) 
    $
    is minimized where
    \[
    d'(x,y) = 
    \begin{cases}
    \max\{d(x,y)-r_i,0\} & \text{ if } x\in P\setminus \{c_1,\ldots,c_{\ell}\},\ y=c_i \text{ with } i \in \{1,\ldots, \ell\}\\
    \max\{d(x,y)-r_i-r_j,0\} & \text{ if } x=c_i,\ y=c_j \text{ with } i,j\in \{1,\ldots, \ell\}\\
    d(x,y) & \text { else}
    \end{cases}
    \]
    i.e., the radius of clusters $1$ to $\ell$ is reduced by $r_1,\ldots,r_{\ell}$ when computing the objective function.
\end{definition}

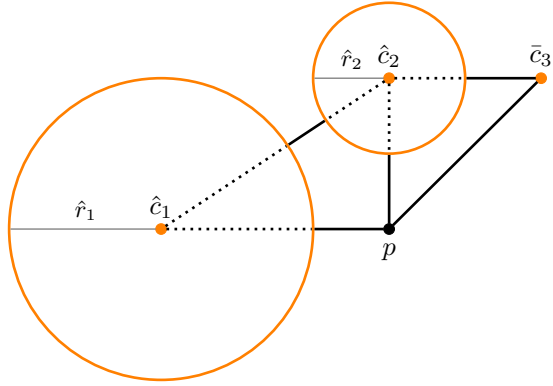
\begin{figure}
    \centering
    \begin{tikzpicture}[scale=1, every node/.style={scale=1}]
        \def \linecolor  {black}
        \def \pointwidth  {2pt}
        \def \linewidth  {1}
        \def \clustercolor {orange}
        
        \def \radiusone {2cm}
        \def \radiustwo {1cm}
        
        \coordinate (1) at (-3,0);
        \coordinate (2) at (0,2);
        \coordinate (p1) at (0,0);
        \coordinate (3) at (2,2);
        
        \draw[fill=black, draw=black] (p1) circle (\pointwidth) {};
        \node at ($(p1)+(0,-.3)$) {$p$};
    
        \foreach \p in {1,2,3}
        {
            \draw[fill=orange, draw=orange] (\p) circle (\pointwidth) {};
        }
        \node at ($(1)+(0,.3)$) {$\hat{c}_1$};
        \node at ($(2)+(0,.3)$) {$\hat{c}_2$};
        \node at ($(3)+(0,.3)$) {$\bar{c}_3$};
        
        \draw[draw=\clustercolor, line width=\linewidth] (1) circle (\radiusone) {};
        \draw[draw=gray, line width=.5\linewidth, shorten <= \pointwidth] (1) -- node[above] {\small $\hat{r}_1$} ($(1)!\radiusone-.5\linewidth!(-10,0)$);
        \draw[draw=\clustercolor, line width=\linewidth] (2) circle (\radiustwo) {};
        \draw[draw=gray, line width=.5\linewidth, shorten <= \pointwidth] (2) -- node[above] {\small $\hat{r}_2$} ($(2)!\radiustwo-.5\linewidth!(-10,2)$);
        
        \draw[-, \linecolor, line width=\linewidth, shorten >=(\radiusone + .5\linewidth), shorten <=\pointwidth] (p1) -- (1);
        \draw[dotted, \linecolor, line width=\linewidth, shorten <=\pointwidth, shorten >=.5\linewidth] (1) -- ($(1)!\radiusone!(p1)$);
        
        \draw[-, \linecolor, line width=\linewidth, shorten >=(\radiustwo + .5\linewidth), shorten <=\pointwidth] (p1) -- (2);
        \draw[dotted, \linecolor, line width=\linewidth, shorten <=\pointwidth, shorten >=.5\linewidth] (2) -- ($(2)!\radiustwo!(p1)$);
        
        \draw[-, \linecolor, line width=\linewidth, shorten <=\pointwidth, shorten >=\pointwidth] (p1) -- (3);
    
        \draw[-, \linecolor, line width=\linewidth, shorten >=(\radiustwo + .5\linewidth), shorten <=\pointwidth] (3) -- (2);
        \draw[dotted, \linecolor, line width=\linewidth, shorten <=\pointwidth, shorten >=.5\linewidth] (2) -- ($(2)!\radiustwo!(3)$);
        
        \draw[-, \linecolor, line width=\linewidth, shorten >=(\radiustwo + .5\linewidth), shorten <=(\radiusone + .5\linewidth)] (1) -- (2);
        \draw[dotted, \linecolor, line width=\linewidth, shorten <=\pointwidth, shorten >=.5\linewidth] (1) -- ($(1)!\radiusone!(2)$);
        \draw[dotted, \linecolor, line width=\linewidth, shorten <=\pointwidth, shorten >=.5\linewidth] (2) -- ($(2)!\radiustwo!(1)$);
    \end{tikzpicture}
    \caption{An instance of a 3-center completion problem. The centers $\hat{c}_1$ and $\hat{c}_2$ with corresponding radii $\hat{r}_1 = 1$ and $\hat{r}_2 = 0.5$ are already given. The underlying distances are given by $d(p,\hat{c}_1) = 1.5$, $d(p,\hat{c}_2) = 1$, $d(p,\bar{c}_3)=\sqrt{2}$. In $d'$, all distances to one of the centers $\hat{c}_1,\hat{c}_2$ are shortened by the respective radius $\hat{r}_1,\hat{r}_2$. Dotted parts indicate the segments that do not contribute to the distance $d'$. For example, $d'(p,\hat{c}_1) = 0.5$, $d'(p,\hat{c}_2) = 0.5$. However, distances not involving $\hat{c}_1$ or $\hat{c}_2$ as one of the end points stay the same, i.e., $d'(p,\bar{c}_3)=d(p,\bar{c}_3)$. Originally, the point $p$ is closer to $\bar{c}_3$ than to $\hat{c}_1$. But under $d'$, $p$ is closer to $\hat{c}_1$.
    This example also shows that the distance $d'$ does not fulfill the triangle inequality: While $d'(p,\bar{c}_3)=\sqrt{2}$, the detour via $\hat{c}_2$ is shorter: $d'(p,\hat{c}_2) + d(\hat{c}_2,\bar{c}_3) = 1$.}
    \label{fig:completion-instance}
\end{figure}

Figure~\ref{fig:completion-instance} shows an example instance for such a completion problem.
The \kcenterproblem can be solved by running an adapted farthest-first traversal starting with $c_1,\ldots,c_{\ell}$ and using distance function $d'$. 
For $i=\ell+1,\ldots,k$, always pick a point $x$ that maximizes $\min_{j \in [i]} d'(x,c_j)$ and set $c_i\coloneqq x$. When started with $\ell=0$ and $d$ instead of $d'$, this is known as farthest-first traversal or Gonzalez' algorithm \cite{G85}. The change is that we already have chosen the first $\ell$ centers and thus they differ from what Gonzalez' algorithm would have picked (and consequently, the remaining centers also differ), and also that the distance to the first $\ell$ points is not metric. 
The adapted algorithm is described in \Cref{alg:gonzalez}. \Cref{lemma:fft-somewhat-different} verifies that the algorithm still succeeds in computing a $2$-approximation.

\begin{algorithm} 
    \LinesNumbered
    \caption{\textsc{Farthest-first-traversal-completion}} \label{alg:gonzalez}
    \SetKwInOut{Input}{Input}
    \SetKwInOut{Output}{Output}
    \BlankLine
    \Input{Point set $P$, distance function $d$, integer $k$, centers $c_1,\ldots,c_i$, radii $r_1,\ldots,r_i$}
    \Output{Set of $k$ centers, assignment $\alpha$}
    \BlankLine
        \CommentSty{Update distance function}
        $d' \gets d$\\
        \For{$j=1,\ldots,i$}{
            \For{$p\in P$}{
                $d'(c_j,p) \gets \max\{d(c_j,p)-r_j,0\}$
            }
        }
        \CommentSty{Complete centers by farthest-first-traversal}
        \For{$j=i+1,\ldots,k$}{
            $c_{i+1} \gets \arg\max_{p\in P}\max_{c\in \{c_1,\ldots,c_i\}}d'(p,c)$
        }
        \CommentSty{Assign points to their closest centers}
        \For{$p\in P$}{
            $\alpha(p) \gets \arg\min_{c\in\{c_1,\ldots,c_k\}}d'(p,c)$
        }
        \Return $c_1,\ldots,c_k$, $\alpha$
\end{algorithm}

\begin{lemma}
\label{lemma:fft-somewhat-different}
    Running Algorithm~\ref{alg:gonzalez} with $d'$ and $c_1, \ldots, c_{\ell}$ already fixed yields a $2$-approximation for the $k$-center completion problem with input $(P,d,k, c_1, \ldots, c_{\ell},r_1,\ldots,r_{\ell}$).
\end{lemma}
\begin{proof}
    We follow the proof for the approximation guarantee of Gonzalez' algorithm and verify that it still works in the case of the somewhat different $k$-center problem. 
    Let $D$ be the maximum distance of any point to its closest point in $\{c_1,\ldots,c_k\}$ with respect to $d'$, i.e., $D$ is the cost of the solution computed by \textsc{Farthest-first-traversal-completion} with $c_1, \ldots, c_{\ell}$ already fixed. Let $c_{k+1}$ be a point with $\min_{i \in [k]} d'(c_{k+1},c_i) = D$. Observe that all $c_i$ with $i \ge \ell +1$ satisfy that $d'(c_i,c_j)\ge D$ for all $j\in\{1,\ldots,\ell\}$ because otherwise $c_{k+1}$ would have been chosen as a center since its minimum distance is $D$. Inductively we also get for all $c_i, c_j$ with $i,j \ge \ell+1$ that $d'(c_i,c_j)\ge D$ is true because otherwise $c_{k+1}$ would have been chosen. Now we get to the point where the proof differs slightly from the original proof because we have a case distinction. We have $k+1$ points $c_1,\ldots,c_{k+1}$. In an optimum solution for the somewhat different $k$-center problem, we can assume that every point is assigned to its closest center, and in particular, all centers are assigned to themselves. There is always an optimum solution that satisfies this (this property is not necessarily ensured for clustering problems with constraints). Let $c_{\ell+1}^{\ast},\dots,c_k^{\ast}$ and $\alpha^{\ast}\colon P \to \{c_1,\ldots,c_{\ell},c_{\ell+1}^{\ast},\ldots,c_k^\ast\}$ be a such an optimal solution, i.e., $\alpha^{\ast}(c_i)=c_i$ for all $i \in [\ell]$.

    Case 1 is that for some $i\in \{\ell+1,\ldots,k+1\}$,  $\alpha^*(c_i)=c_j \in \{c_1,\ldots,c_\ell\}$, i.e., one of the points we picked as a center or the additional point $c_{k+1}$ is in the optimum solution assigned to one of the predefined centers. In this case, $OPT \ge d'(c_i,c_j) \ge D$ since we argued that all our centers have distance of at least $D$ to the predefined centers.

    Case 2 is that none of $c_{\ell+1},\ldots,c_{k+1}$ is assigned to any predefined center. Thus, they are all assigned to the $k-\ell$ centers $c_{\ell+1}^{\ast},\ldots,c_k^{\ast}$. By the pigeonhole principle, this means that $\alpha^{\ast}(c_i) = \alpha^{\ast}(c_j)=c_m^\ast$ for some $i, j \in \{\ell+1,\ldots,k+1\}$ and $m\in \{\ell+1,\ldots,k\}$. Since $i,j \in [\ell]$, $d'(c_i,c_j) \ge D$ as argued above. Also since $i,j,m \in [\ell]$, by definition, $d'(c_i,c_j)=d(c_i,c_j)$, $d'(c_i,c_m)=d(c_i,c_m)$  and $d'(c_j,c_m)=d(c_j,c_m)$. 
       
    We conclude by the triangle inequality that
    \begin{align*}
    D \leq d'(c_i, c_j) = d(c_i, c_j) &\leq d(c_i, c_m^{\ast}) + d(d_m^{\ast}, c_j) \\
    &\leq 2\max_{g=i,j}\{d(c_g, \alpha^{\ast}(c_g))\} \leq 2\max_{g\geq \ell+1}\{d(c_g, \alpha^{\ast}(c_g))\}.
    \end{align*}
    As $d(c_g, \alpha^*(c_g)) = d'(c_g, \alpha^*(c_g))$ for all $g\geq \ell+1$ by definition, this implies $\OPT \geq \frac{1}{2}D$.
\end{proof}

\subsection{Guessing an Approximate Radius Profile for the Optimum Solution}

In the full version, we obtain the following corollary that allows us to guess close approximations for all radii of the optimal $k$-MSR solution in FPT time.
Let $(r_1^*, \ldots, r_k^*)$ be the radius profile of an optimal solution. We call a radius profile $(\tilde{r}_1, \ldots, \tilde{r}_k)$ \emph{near-optimal} if $r_i^* \leq \tilde{r}_i \leq (1+\eps) r_i^*$ for all $i\in \{1, \ldots, k\}$.

\begin{restatable}{corollary}{computeradiusprofile}\label{cor:compute-radius-profile}
Let $(r_1^*, \ldots, r_k^*)$ be the radius profile of an optimal solution, and assume that we know the value of a constant-factor approximation solution for the corresponding $k$-center problem on the same instance. Then we can compute a set of size $O(\log^k_{1+\eps}(k/\eps))$ that contains a near-optimal radius profile $(\tilde{r}_1, \ldots, \tilde{r}_k)$ in time $O(k\log^k_{1+\eps}(k/\eps))$.
\end{restatable}

\section{\texorpdfstring{Algorithm for $k$-Min-Sum-Radii with Mergeable Constraints}{Algorithm for k-Min-Sum-Radii with Mergeable Constraints}}

The aim of this section is to prove the following Theorem~\ref{maintheorem} which is proven in Section~\ref{subsec:assighment}, followed by Theorem~\ref{thm-1-1} for 1:1 fairness in \Cref{sec:1-1-fairness} and Corollary~\ref{cor-lower} for lower bounds in \Cref{sec:lower-bounds}.

\begin{restatable}{theorem}{maintheorem}\label{maintheorem}
    For every $\eps > 0 $, there exists an algorithm that computes a $(6-\frac{3}{k}+\eps)$-approximation for $k$-min-sum-radii with mergeable constraints in time $O((k \log_{1+\eps}(k/\eps))^k\cdot \poly(n))$ if the corresponding constrained $k$-center problem has a constant factor approximation algorithm.
\end{restatable}
\vspace{1em}
\noindent Our algorithm works in two steps. 
First, the algorithm computes a candidate set of $k$ radii and centers based on guessing. If it guesses correctly, the induced balls form a feasible $k$-min-sum-radii solution with certain properties. However, it might not fulfill the mergeable constraint yet. What it means to guess correctly is defined later in \Cref{def:guessing-correctly} after the algorithm is specified.
Notice that for this part, we need no assumptions about the constraint aside from the fact that an approximation algorithm for the $k$-center problem under this mergeable constraint exists. 
We need mergeability only in the computation of the final assignment in Section~\ref{subsec:assighment}. 

In the second step of the algorithm, we compute an assignment of points to the candidate centers. If the center and radius candidates from the first step are appropriate, then this assignment is guaranteed to fulfill the mergeable constraint.  

In the following, we fix an optimal solution that we are trying to find. It consists of clusters $C_1^*, \ldots, C_k^*$, with centers $c_1^*, \ldots, c_k^*$ and radii $r_1^*, \ldots, r_k^*$. We will assume that the optimal radii are sorted decreasingly. All clusters necessarily fulfill the mergeable constraint. Furthermore, we assume that we are in an iteration where we consider the radius profile $\tilde{r}_1, \ldots, \tilde{r}_k$ satisfying $r_j^* \leq \tilde{r}_j \leq (1+\eps)r_j^*$ for all $j\leq k$. We also say that such a radius profile is near-optimal. Such an iteration exists due to \Cref{cor:compute-radius-profile}. During this run of the algorithm, we are constructing candidate centers $\hat{c}_1,\ldots,\hat{c}_{k}$ and candidate radii $\hat{r}_1, \ldots, \hat{r}_k$. In summary, we have the following notation to be aware of during the following:
\begin{itemize}
    \item $C_1^*, \ldots, C_k^*$ denote an optimal clustering for $k$-MSR with mergeable constraint with centers $c_1^*, \ldots, c_k^*$ and radii $r_1^*\geq \ldots\geq r_k^*$
    \item $\tilde{r}_1, \ldots, \tilde{r}_k$ denote initial near-optimal radius profile such that $r_j^* \leq \tilde{r}_j \leq (1+\eps)r_j^*$ for all $j\leq k$ 
    \item $\hat{c}_1, \ldots, \hat{c}_i$, $\hat{r}_1, \ldots, \hat{r}_i$ denote candidate centers and radii constructed up to iteration $i$
\end{itemize}

\subsection{Selection of Candidate Centers and Radii} \label{sec:initial-selection-of-centers-and-radii}
The general idea of the algorithm is as follows:
Assume that in the beginning of iteration $i$, we already fixed candidate centers $\hat{c}_1,\ldots,\hat{c}_{i-1}$ and candidate radii $\hat{r}_1,\ldots,\hat{r}_{i-1}$. We compute a 2-approximation for the induced $k$-center completion instance. 
The resulting output consists of centers $\hat{c}_1,\ldots,\hat{c}_{i-1},\bar{c}_i,\ldots,\bar{c}_{k}$, radii $\hat{r}_1,\ldots,\hat{r}_{i-1},\bar{r}_i,\ldots,\bar{r}_{k}$ and an assignment $\alpha$. 
We guess $\alpha(c_i^*)$, i.e.\ where the $i$-th center of the optimal solution is assigned to in the $k$-center completion solution. Recall that we already have a good approximation for $\tilde{r_i}$ for the corresponding optimal solution radius $r_i^*$.
If we guess that $\alpha(c_i^*)$ is among the newly chosen centers, i.e.\ if $ \alpha(c_i^*) = \bar{c}_j \in \{\bar{c}_i,\ldots,\bar{c}_{k}\}$, we open a new ball with radius $\hat{r}_i = 3\tilde{r}_i$ at this center $\hat{c}_i \coloneqq \bar{c}_j$. 
Otherwise, if we guess that $\alpha(c_i^*) = \hat{c}_j \in \{\hat{c}_1,\ldots,\hat{c}_{i-1}\} $, there already exists a ball around this center, and we only need to enlarge this ball by $3\tilde{r}_i$. In order to have $k$ centers in the end, we set $\hat{c}_i$ to some arbitrary point and $\hat{r}_i = 0$.

The guessing of the center assignments can be handled as follows. In every iteration $i$, we have $k$ possible choices for $\alpha(c^*_i)$. So each sequence of $k$ ``guesses'' can be encoded by a tuple $a = (a_1, \ldots, a_k)\in \{1, \ldots, k\}^{k}$, where $a_i = \ell$ means that in the $i$th iteration, we choose the $\ell$-th center of the $k$-center completion solution as $\alpha(c_i^*)$ (i.e. $\hat{c}_\ell$ if $\ell < i$, or $\bar{c}_\ell$ if $\ell \geq i$). Thus, we can emulate the guessing by generating all such tuples upfront, computing the candidate balls for each of these, and choosing the best feasible one in the end.
For a formal description of the algorithm, see \Cref{alg:centers-and-radii}.

\begin{algorithm} 
    \LinesNumbered
    \caption{\textsc{Centers-and-radii}}
    \label{alg:centers-and-radii}
    \SetKwInOut{Input}{Input}
    \SetKwInOut{Output}{Output}
    \BlankLine
    \Input{Points $P$, distances $d$, $k\in\mathbb{N}$, radius profile $(\tilde{r_1}, \ldots, \tilde{r_k})$, tuple $(a_1, \ldots, a_k)$}
    \Output{Set of $k$ centers, set of $k$ radii}
    \BlankLine
        $I_0 \gets (P,d,k,\emptyset, \emptyset)$\\
        \For{$i=1,\ldots,k$}{
            $(S_{kcc},\alpha) \gets $ \textsc{Farthest-first-traversal-completion}($I_{i-1}$) , where $S_{kcc} = \{\hat{c}_1,\ldots,\hat{c}_{i-1}, \bar{c}_i,\ldots, \bar{c}_k\}$ and $\alpha\colon P\to S_{kcc}$ \label{alg-line:call-farthest-first-traversal}\\
            \If{$a_i < i$}{
            \CommentSty{guess that $\alpha(c_i^*) \in \{\hat{c}_1, \ldots, \hat{c}_{i-1}\}$}
                Set $\hat{r}_{a_i} \gets \hat{r}_{a_i} + 3\tilde{r}_{i}$, choose $\hat{c}_i$ arbitrarily, set $\hat{r}_i \gets 0$\label{alg-line:enlarge-ball}\\
            }
            \ElseIf{$a_i \geq i$}{
            \CommentSty{guess that $\alpha(c_i^*) \in \{\bar{c}_i, \ldots, \bar{c}_{k}\}$}
                Set $\hat{c}_i \gets \bar{c}_{a_i}, \hat{r}_{i} \gets 3\tilde{r}_{i}$ \label{alg-line:new-ball}\\
            }
            $I_{i} \gets (P,d,k,\{\hat{c}_1,\ldots,\hat{c}_{i}\},\{\hat{r}_1,\ldots,\hat{r}_{i}\})$\\
        }
        \Return $\{\hat{c}_1,\ldots,\hat{c}_k\}$, $\{\hat{r}_1,\ldots,\hat{r}_k\}$
\end{algorithm}

In the full version, we show an example run of \Cref{alg:centers-and-radii}.
With this notation and \Cref{alg:centers-and-radii} in place, we can now formally define what it means to guess correctly.

\begin{definition}[Guessing correctly] \label{def:guessing-correctly}
    Given a solution $(S_{kcc},\alpha)$ for the $k$-center completion problem with input $\hat{c}_1,\ldots,\hat{c}_{i-1}$ and $\hat{r}_1,\ldots,\hat{r}_{i-1}$.
    We say that \Cref{alg:centers-and-radii} guesses correctly if the input tuple $a$ is such that in every iteration $i$, $a_i$ is a correct guess of the assignment of the next optimal center under $\alpha$. To be more precise, $a_i$ is the smallest index in $\{1,\ldots,k\}$ such that $ c_{a_i} = \alpha(c_i^*) $ with $ S_{kcc} = \{c_1,\ldots, c_k\}$, where $c_i^*$ is the center of the next optimal cluster $C_i^*$.
\end{definition}

The idea of \Cref{alg:centers-and-radii} is that it fully covers one so-far uncovered optimal cluster in every iteration (under the assumption that the initial radius profile is near-optimal and \Cref{alg:centers-and-radii} guesses correctly). For the analysis, we need the following Lemma that bounds the cost of an optimal $k$-center completion solution in any iteration of \Cref{alg:centers-and-radii} by the radius of the largest remaining optimal cluster that is not fully covered yet. Combining with \Cref{lemma:fft-somewhat-different} gives an upper bound on the distance between an optimal center $c^*$ and the center $\alpha(c^*)$ it is assigned to. 

\begin{lemma} \label{lem:k-center-completion-upper-bound-r*}
    Assume that up to the end of iteration $i$, \Cref{alg:centers-and-radii} chose centers $\hat{c}_1, \ldots, \hat{c}_i$ and radii $\hat{r}_1, \ldots, \hat{r}_i$ such that for all $p\in \bigcup_{j\leq i}C_j^*$, there exists a center $\hat{c}_\ell \in \{\hat{c}_1, \ldots, \hat{c}_i\}$ such that $d'(p, \hat{c}_{\ell}) = 0$. Then, the value of an optimal solution for the \kcenterproblem with input $\{\hat{c}_1,\ldots,\hat{c}_i\}, \{\hat{r}_1,\ldots,\hat{r}_i\}$ is at most $r_{i+1}^*$. 
\end{lemma}
\begin{proof}
    Consider the center extension  $\{c^*_{i+1}, \ldots, c^*_k\}$. 
    By the precondition of the lemma, we can assign every point in $p\in \bigcup_{j\leq i}C_j^*$ to a center in $\{\hat{c}_1, \ldots, \hat{c}_i\}$ at distance $0$ with respect to $d'$. For all $h \ge i+1$, any $x \in C_h^*$ can be assigned to $c_h^*$ at distance $\le r_h^*$. As the optimal radii are sorted in decreasing order, $r_{i+1}^*$ is the largest remaining radius among the optimal clusters under $d'$. Hence, the resulting assignment $\alpha'$ satisfies $d'(x,\alpha'(x))\le r_{i+1}^*$. Notice that $\{c^*_{i+1}, \ldots, c^*_k\}$ and $\alpha'$ form 
    a feasible solution for the \kcenterproblem and that the maximum radius of this solution is $r_{i+1}^*$ as argued above.
    Hence, the optimum value for the \kcenterproblem is upper bounded by $r^*_{i+1}$. 
\end{proof}

Now, we can show that there exists a surjective mapping $\varphi\colon \{C_1^*,\ldots,C_k^*\} \to \hat{\mathcal{B}} $, where $\hat{\mathcal{B}}$ is the collection of balls from $\{B(\hat{c}_1,\hat{r}_1),\ldots,B(\hat{c}_k, \hat{r}_k)\}$ for which $\hat{r}_i > 0$, such that $ C_j^*\subseteq \varphi(C_j^*)$ for all $j\leq k$.
The next Lemma formalizes this. 
 
\begin{lemma}\label{lem:injective-mapping}
    Assume that our guess of the initial radius profile is near-optimal and that \Cref{alg:centers-and-radii} guesses correctly.
    Let $\hat{\mathcal{B}}$ denote the set of balls $B(\hat{c}_1,\hat{r}_1), \ldots, B(\hat{c}_k,\hat{r}_k)$ found by \Cref{alg:centers-and-radii}. Then the following two statements hold true
    \begin{enumerate}
        \item for all $j\leq k$, there exists $\ell\leq k$ such that $C_j^*\subseteq B(\hat{c}_{\ell},\hat{r}_{\ell})$
        \item for all $\ell \leq k$, if $r_{\ell}>0$ then there exists $j\leq k$ such that $C_j^*\subseteq B(\hat{c}_{\ell}, \hat{r}_{\ell})$. 
    \end{enumerate}
\end{lemma} 
\begin{proof}
    We show that the statement holds at the end of every iteration of \Cref{alg:centers-and-radii}. That is, we show that for all $i\leq k$, the following holds
    \begin{enumerate}[(1)]
        \item for all $j\leq i$, there exists $\ell\leq i$ such that $C_j^*\subseteq B(\hat{c}_{\ell},\hat{r}_{\ell})$
        \item for all $\ell \leq i$, if $r_{\ell}>0$ then there exists $j\leq i$ such that $C_j^*\subseteq B(\hat{c}_{\ell}, \hat{r}_{\ell})$. 
    \end{enumerate}
    Then setting $i=k$ implies the result.
    We prove this via induction over $i\leq k$.
    For $i=0$, the statements are trivially fulfilled.
    Now let the statement be fulfilled at the end of iteration $i-1< k$ for some $i>0$. 
    
    By \Cref{lem:k-center-completion-upper-bound-r*}, $\OPT_{kcc} \leq r_i^*$, where $\OPT_{kcc}$ is the value of an optimal solution for the $k$-center completion problem that takes the centers and radii generated until the end of iteration $i-1$ as input. By \Cref{lemma:fft-somewhat-different}, \textsc{Farthest-first-traversal-completion} in Line~\ref{alg-line:call-farthest-first-traversal} of \Cref{alg:centers-and-radii} computes a 2-approximation for the $k$-center completion problem. These two arguments together imply $d(c_i^*,\alpha(c_i^*)) \leq 2\OPT_{kcc} \leq 2r_i^*$.
    
    If $c_i^*=\hat{c}_j$ for some $j\leq i-1$, then for all $p\in C_i^*$, it is $d(p,\hat{c}_j) = d(p,c_i^*) \leq r_i^* \leq r_j^*$, where the last inequality holds because the optimal radii are sorted decreasingly. 
    
    If \Cref{alg:centers-and-radii} guesses correctly, $c_{a_i} = \alpha(c_i^*) = c_i^* = \hat{c}_j$ and the radius $\hat{r}_j^{\text{new}}$ produced in Line~\ref{alg-line:enlarge-ball} fulfills $\hat{r}_j^{\text{new}} \coloneqq \hat{r}_j + 3\tilde{r}_i \geq r_i^*$. Therefore $C_i^*$ is covered completely by $B(\hat{c}_j,\hat{r}_j^{\text{new}})$. This implies that (1) holds. For index $i$, the algorithm creates a new ball with radius $\hat{r}_i = 0$. Therefore, statement (2) is fulfilled by the induction hypothesis.

    Now, we assume that $c_i^*\notin \{\hat{c}_1,\ldots,\hat{c}_{i-1}\}$.
    There are two cases for the guess $a_i$.
    \begin{enumerate}
        \item Either $a_i < i$. Then, we are in Line~\ref{alg-line:enlarge-ball} of \Cref{alg:centers-and-radii} and enlarge an already existing ball centered at $\alpha(c_i^*) = \hat{c}_{a_i}$ by $3\tilde{r}_i$, i.e., the $a_i$-th ball is $B(\hat{c}_{a_i},\hat{r}_{a_i} + 3\tilde{r}_i)$ at the end of the iteration. 
        For every $p\in C_i^*$,
        \[ d(p,\hat{c}_{a_i}) \leq d(p,c_i^*) + d(c_i^*,\hat{c}_{a_i}) =d(p,c_i^*) + d'(c_i^*,\hat{c}_{a_i}) + \hat{r}_{a_i} = d(p,c_i^*) + d'(c_i^*,\alpha(c_i^*)) + \hat{r}_{a_i}. \]
        It is $d(p,c_i^*) \leq r_i^*$ as $p\in C_i^*$, and $d'(c_i^*,\alpha(c_i^*)) \leq 2r_i^*$ as $\alpha$ is the assignment given by the 2-approximation. Further, $r_i^* \leq \tilde{r}_i$. Overall, $d(p,\hat{c}_{a_i}) \leq 3\tilde{r}_i + \hat{r}_{a_i}$, which implies that the ball $B(\hat{c}_{a_i},\hat{r}_{a_i} + 3\tilde{r}_i)$ covers $C_i^*$ completely.
        \item Or $a_i \geq i$. In this case, the algorithm creates a new ball $B(\hat{c}_i,\hat{r}_i)$ with $\hat{c}_i \coloneqq c_{a_i} = \alpha(c_i^*)$ and $\hat{r}\coloneqq 3\tilde{r}_i$.
        For every $p\in C_i^*$,
        \[ d(p,\hat{c}_i) = d(p,\alpha(c_i^*)) \leq d(p,c_i^*) + d(c_i^*,\alpha(c_i^*)) = d(p,c_i^*) + d'(c_i^*,\alpha(c_i^*)), \]
        where the last equality holds because $c_i^*\not\in \{\hat{c}_1,\ldots,\hat{c}_{i-1}\}$ and $a_i$ is the smallest index such that $c_{a_i} = \alpha(c_i^*)$, which implies $c_{a_i} \not\in \{\hat{c}_1,\ldots,\hat{c}_{i-1}\}$.
        Again, $d(p,c_i^*) \leq r_i^*$ and $d'(c_i^*,\alpha(c_i^*)) \leq 2r_i^*$. Overall, $d(p,\hat{c}_i) \leq 3r_i^* \leq 3\tilde{r}_i$, and hence, $C_i^*$ is completely covered by $B(\hat{c}_i,\hat{r}_i)$.
    \end{enumerate}
\end{proof}
Notice that the candidate balls might overlap, but the optimal clusters are pairwise disjoint by definition of a clustering. The following lemma relates the total cost of the clustering consisting of the candidate balls to the cost of an optimal $k$-min-sum-radii with mergeable constraints solution. This will be useful for analyzing the cost of our final solution later on. Notice that this statement does not imply an approximation ratio for the vanilla $k$-min-sum-radii problem.

\begin{lemma}\label{lem:algo-apx}
    Let $\hat{r}_1,\ldots,\hat{r}_k$ be the radii produced by Alg.~\ref{alg:centers-and-radii}. Then $\sum_{j=1}^k \hat{r}_j \leq 3(1+\eps)\sum_{j=1}^k r_j^*$.
\end{lemma}
\begin{proof}
    We show by induction that $\sum_{j=1}^i \hat{r}_j \leq 3\cdot \sum_{j=1}^i \tilde{r}_j$ for all $i\leq k$. Then for $i=k$, the result follows since $\tilde{r}_j \leq (1+\eps)r_j^*$ for all $j\leq k$.

    For $i=0$, the statement trivially holds.
    Now assume that the statement holds for $i-1$.
    Either the algorithm sets $\hat{r}_i \coloneqq 3\tilde{r}_i$ during iteration $i$. Then, 
    $ \sum_{j=1}^i \hat{r}_j = \sum_{j=1}^{i-1}\hat{r}_j + \hat{r}_i \leq 3\sum_{j=1}^{i-1}\tilde{r}_j + 3\tilde{r}_i$.
    For the remaining case, let $\hat{r}_j^{(i-1)}$ denote the value of $\hat{r}_j$ at the beginning of the $i$th iteration, and $\hat{r}_j^{(i)}$ its value at the end of the iteration, $j\leq i$.
    There exists $\ell <i$ such that the algorithm sets $\hat{r}_\ell^{(i)} \coloneqq \hat{r}_\ell^{(i-1)} + 3\tilde{r}_i$ and $\hat{r}_i^{(i)} \coloneqq 0$. Then,
    $ \sum_{j=1}^i \hat{r}_j^{(i)} = \sum_{j=1}^{i-1} \hat{r}_j^{(i)} = \sum_{j=1}^{i-1} \hat{r}_j^{(i-1)} + 3\tilde{r}_i \leq 3\cdot \sum_{j=1}^i \tilde{r}_j. $
\end{proof}

\subsection{Finding the Assignment}
\label{subsec:assighment}
In the following, we will show how to find a feasible assignment. We construct a graph from center and radii candidates computed in the first part of the algorithm and observe that the clustering induced by the connected components of this graph fulfills the given mergeable constraint. We define the \emph{access graph} $G=(V,E)$ as follows. The set of vertices corresponds to the given point set, i.e.\ $V=P$. We add an edge between any pair of vertices $x,y\in V$ iff $x=\hat{c}_i$ for a center $\hat{c}_i$ constructed in Algorithm~\ref{alg:centers-and-radii} and $d(y,\hat{c}_i) \leq \hat{r}_i$ for the corresponding radius $\hat{r}_i$. The construction is exemplified in \Cref{fig:graph-construction}.
A connected component is a maximal connected subgraph of $G$. Let $\CC(G)$ denote the set of connected components in $G$.
Covering a connected component $Z\in \CC(G)$ using one large cluster is not more expensive than covering it using the balls $B(\hat{c},\hat{r})$ for all $\hat{c}\in Z$.

\begin{figure}
    \centering
    \begin{tikzpicture}
        \clip (-2,-3) rectangle (8,2);
        
        \def \linecolor  {black}
        \def \centerwidth  {2pt}
        \def \pointwidth  {1.2pt}
        \def \linewidth  {.5pt}
        
        \coordinate (c1) at (0,-0.3);
        \coordinate (c2) at (2,0);
        \coordinate (c3) at (3.8,-0.7);
        \coordinate (c4) at (7,-1.4);

        \node at ($(c1) + (0.25,-0.25)$) {$\hat{c}_1$};
        \node at ($(c2) + (0,0.3)$) {$\hat{c}_2$};
        \node at ($(c3) + (0,0.3)$) {$\hat{c}_3$};
        \node at ($(c4) + (-0.25,0.25)$) {$\hat{c}_4$};
        
        \coordinate (11) at ($ (c1) + (0.1,.6) $);
        \coordinate (12) at ($ (c1) + (1.2,.4) $);
        \coordinate (13) at ($ (c1) + (-.5,.3) $);
        \coordinate (14) at ($ (c1) + (-.3,-.5) $);
        \coordinate (15) at ($ (c1) + (-1,0) $);
        \coordinate (16) at ($ (c1) + (-.4,1) $);
        
        \coordinate (21) at ($ (c2) + (.8,-.2) $);
        
        \coordinate (31) at ($ (c3) + (.8,-.2) $);
        \coordinate (32) at ($ (c3) + (.4,-.9) $);
        
        \coordinate (41) at ($ (c4) + (.3,.2) $);
        \coordinate (42) at ($ (c4) + (0,-.4) $);
        
        \foreach \p in {c1,c4,16,12,31}
        {
            \draw[fill=orange, draw=orange] (\p) circle (\pointwidth) {};
        }
        \foreach \p in {c2,c3,11,13,14,15,21,32,41,42}
        {
            \draw[fill=blue, draw=blue] (\p) circle (\pointwidth) {};
        }
			
        \foreach \p in {c1,c4}
        {
            \draw[fill=orange, draw=orange] (\p) circle (2pt) {};
        }
        \foreach \p in {c2,c3}
        {
            \draw[fill=blue, draw=blue] (\p) circle (2pt) {};
        }

        \draw[draw=black] (c1) circle (1.5) {};
        \draw[draw=black] (c2) circle (1) {};
        \draw[draw=black] (c3) circle (1.2) {};
        \draw[draw=black] (c4) circle (.8) {};

        \draw[-, gray, line width=\linewidth, shorten <=\centerwidth, shorten >=\linewidth] (c1) -- node[below right] {$\hat{r}_1$} ($(c1)+(0,-1.5)$);
        \draw[-, gray, line width=\linewidth, shorten <=\centerwidth, shorten >=\linewidth] (c2) -- node[right] {$\hat{r}_2$} ($(c2)+(0,-1)$);
        \draw[-, gray, line width=\linewidth, shorten <=\centerwidth, shorten >=\linewidth] (c3) -- node[left] {$\hat{r}_3$} ($(c3)+(0,-1.2)$);
        \draw[-, gray, line width=\linewidth, shorten <=\centerwidth, shorten >=\linewidth] (c4) -- node[below] {$\hat{r}_4$} ($(c4)+(.8,0)$);

        \foreach \p in {11,12,13,14,15,16}
        {
            \draw[-, \linecolor, line width=\linewidth, shorten <=\centerwidth, shorten >=\pointwidth] (c1) -- (\p);
        }
        
        \foreach \p in {21,12}
        {
            \draw[-, \linecolor, line width=\linewidth, shorten <=\centerwidth, shorten >=\pointwidth] (c2) -- (\p);
        }
        
        \foreach \p in {31,32,21}
        {
            \draw[-, \linecolor, line width=\linewidth, shorten <=\centerwidth, shorten >=\pointwidth] (c3) -- (\p);
        }
        
        \foreach \p in {41,42}
        {
            \draw[-, \linecolor, line width=\linewidth, shorten <=\centerwidth, shorten >=\pointwidth] (c4) -- (\p);
        }
\end{tikzpicture}
\caption{An instance of a $k$-min-sum-radii problem with exact fairness constraint with two colors and a blue:orange ratio of 2:1. The larger dots indicate centers and the gray lines indicate the radii output by Alg.~\ref{alg:centers-and-radii}. The black circles show the induced balls $B(\hat{c}_i, \hat{r}_i)$. The black lines between points represent the edges of the induced access graph. Note that the balls themselves are not necessarily fair, but every connected component is.}
\label{fig:graph-construction}
\end{figure}
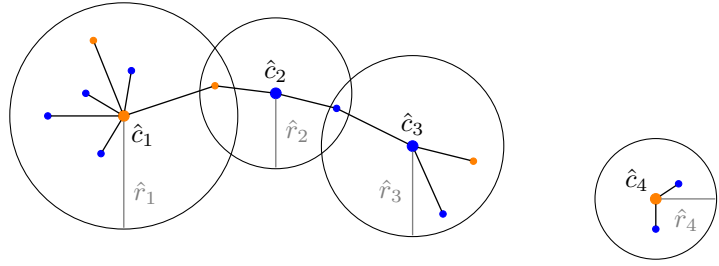

\begin{lemma} \label{lem:connected-component-is-fair}
    Assume Algorithm~\ref{alg:centers-and-radii} made the correct decision in each iteration and terminates with centers $\hat{C} = \left\{\hat{c}_1,\ldots, \hat{c}_k\right\}$ and radii $\hat{r}_1,\ldots,\hat{r}_k$.
    Let $G = (V,E)$ be the corresponding  access graph.
    Let $Z\in \CC(G)$ be a connected component of $G$. Then assigning all vertices from $Z$ to an arbitrary point in $Z$ yields a cluster that is feasible with respect to the given mergeable constraint.
\end{lemma}
\begin{proof}
    Let $Z \in \CC(G)$ be a connected component of $G$. Let $\V(Z)$ denote the set of vertices of $Z$.
    We will show that $\V(Z)$ consists solely of $\ell$ optimal clusters that all lie entirely in $\V(Z)$ for some $\ell\geq 1$.
    As optimal clusters fulfill the mergeable constraint, the union of these also fulfills the mergeable constraint.

    Every point $p\in \V(Z)$ lies in some optimal cluster. Hence, there exists at least one optimal cluster that intersects $\V(Z)$. We want to conclude that such a cluster already is completely contained in $\V(Z)$.
    Assume for a contradiction that there exists an optimal ball $C^*$ such that $C^*\cap \V(Z) \neq \emptyset$ and $C^* \not\subseteq \V(Z)$. 
    By \Cref{lem:injective-mapping}, there exists a ball $B(\hat{c},\hat{r})$ such that $C^* \subseteq B(\hat{c},\hat{r})$. Let $v\in C^*\setminus \V(Z)$. Then $d(v,\hat{c}) \leq \hat{r}$ and therefore $\hat{c}$ must be part of the connected component $Z$, a contradiction.
\end{proof}
We can use this insight as follows: For every connected component $Z\in \CC(G)$, pick one of the centers $\hat{c} \in \hat{C} \cap \V(Z)$ that lie inside the connected component and assign all points in $\V(Z)$ to $\hat{c}$. 
This way, we get a solution that contains one cluster per connected component. To achieve the smallest possible cost guarantee, we set $\hat{c}$ with the largest corresponding $\hat{r}$ as the final center.

\begin{lemma}\label{lem:final-approximation-ratio}
    Let $\hat{C} = \{\hat{c}_1,\ldots,\hat{c}_k\}$, $\hat{r}_1,\ldots,\hat{r}_k$ and $G=(V,E)$ as in \Cref{lem:connected-component-is-fair}. For every connected component $Z$, we choose the center $\cCC{Z}{} \in \hat{C}\cap \V(Z)$ such that the corresponding radius $\rCC{Z}{}$ is maximal among all radii in the connected component. Then, the solution $(\mathscr{C},f)$ with $\mathscr{C} = \{\cCC{Z}{} \mid Z \in \CC(G)\}$ and $f\colon P \to \mathscr{C}$ with $f(z) = \cCC{Z}{}$ for all $z\in Z$ and for all connected components $Z$ is a $(6-\frac{3}{k}+\eps)$-approximation for the $k$-min-sum-radii problem under a mergeable constraint.
\end{lemma}
\begin{proof}
    Since each cluster in the solution corresponds to exactly one connected component of $G$, \Cref{lem:connected-component-is-fair} implies that the solution fulfills the mergeable constraint.

    It remains to prove the approximation factor.
    Let $Z\in \CC(G)$ be a connected component in $G$. 
    Let $v^Z \in \arg\max_{p\in Z}d(\cCC{Z}{},v^Z)$. 
    There exists a path from $\cCC{Z}{}$ to $v^Z$ in $Z$. 
    A shortest such path $\cCC{Z}{}, v^1, \cCC{Z}{1}, v^2, \cCC{Z}{2}, \ldots v^{\ell}, \cCC{Z}{\ell_Z}, v^Z$ with $\ell_Z \leq k$ alternatingly visits points in $\hat{C}$ and $\V(Z)\setminus \hat{C}$. 
    Therefore, its length is bounded by $\sum_{i\leq \ell_Z} 2\rCC{Z}{i} - \rCC{Z}{\max} $, where $\rCC{Z}{\max} \coloneqq \max_{i\colon \hat{c}_i\in Z}\rCC{Z}{i}$.
     The radius of a cluster with center $\cCC{Z}{}$ is given by $d(v^Z,\cCC{Z}{})$. Hence, the sum of the radii of such clusters is bounded by
    \[ \sum_{Z\in \CC(G)} d(v^Z, \cCC{Z}{}) \leq \sum_{Z\in \CC(G)}\biggl(\sum_{i\leq \ell_Z} 2\rCC{Z}{i} - \rCC{Z}{\max}\biggr) = \sum_{j=1}^k 2\hat{r}_j - \sum_{Z\in \CC(G)}\rCC{Z}{\max} \]
    where the second equality holds because a graph's connected components are disjoint. 
    There exists a connected component $Z'$ such that $\rCC{Z'}{\max} = \max_{i\leq k}\hat{r}_i \eqqcolon \hat{r}_{\max}$. Therefore, $\sum_{Z\in \CC(G)}\rCC{Z}{\max} \geq \hat{r}_{\max}$. Further, $\hat{r}_{\max} \geq \frac{1}{k}\sum_{j=1}^k \hat{r}_j$.
    Hence,
    \[ \sum_{j=1}^k 2\hat{r}_j - \sum_{Z\in \CC(G)}\rCC{Z}{\max} \leq \bigl(2- \frac{1}{k}\bigr)\sum_{j=1}^k \hat{r}_j,\] 
    and by \Cref{lem:algo-apx},
    \[ \bigl(2- \frac{1}{k}\bigr)\sum_{j=1}^k \hat{r}_j \leq 3\bigl(1+\eps\bigr)\bigl(2-\frac{1}{k}\bigr)\sum_{j=1}^k r^*_j = \bigl(6-\frac{3}{k}\bigr) \bigl(1+\eps\bigr) \sum_{j=1}^k r^*_j.\]
\end{proof}
\Cref{alg:k-MSR-with-constraints} finds such a solution. Now we are ready to prove our main Theorem.

\begin{algorithm}
    \LinesNumbered
    \caption{\textsc{Assignment}}
    \label{alg:fair-assignment-general-case}
    \SetKwInOut{Input}{Input}
    \SetKwInOut{Output}{Output}
    \BlankLine
    \Input{Graph $G=(V,E)$, distance function $d$, set of centers $\hat{c}_1,\ldots,\hat{c}_k$}
    \Output{Set of $ \leq k$ centers $\mathscr{C}$, assignment $f$}
    \BlankLine
        $ \mathscr{C} \gets \emptyset$\\
        \For{each connected component $Z$ of $G$}{
            find a center $\hat{c} \in Z\cap \hat{C}$ such that $\hat{r}$ is largest\\
            $ \mathscr{C} \gets \mathscr{C} \cup \{\hat{c}\}$\\
            \For{all $p\in Z$}{
                $f(p) \gets \hat{c}$
            }
        }
        \Return $ \mathscr{C} $, $f$
\end{algorithm}

\begin{algorithm}
    \LinesNumbered
    \caption{\textsc{\texorpdfstring{$k$-min-sum-radii with mergeable constraints}{k-min-sum-radii with mergeable constraints}}}
    \label{alg:k-MSR-with-constraints}
    \SetKwInOut{Input}{Input}
    \SetKwInOut{Output}{Output}
    \BlankLine
    \Input{Point set $P$, distance function $d$, $k\in \mathbb{N}$}
    \Output{Set of $ \leq k$ centers $\mathscr{C}$, assignment $f$}
    \BlankLine
        $U \gets (6+\eps)\max_{x,y\in P}d(x,y)$\hfill\CommentSty{upper bound on the sum of radii cost}\
        $\mathscr{R} \gets \text{set of radius profile guesses}$\label{alg-line:radius-guesses}\\
        \ForAll{$(\tilde{r_1}, \ldots, \tilde{r_k})\in \mathscr{R}$}{
            \ForAll{$a \in \{1,\ldots, k\}^k$}{
                $(\{\hat{c}_1,\ldots, \hat{c}_k\}, \{\hat{r}_1,\ldots, \hat{r}_k\}) \gets  \textsc{Centers-and-radii}(P, d, k, (\tilde{r_1}, \ldots, \tilde{r_k}), a$) \label{alg-line:call-centers-and-radii}\\
                compute the access graph $G$ based on $\{\hat{c}_1,\ldots, \hat{c}_k\} $ and $ \{\hat{r}_1,\ldots, \hat{r}_k\}$\\
                $ (\mathscr{C},f) \gets $ \textsc{Assignment}($G$, $d$, $\{\hat{c}_1,\ldots, \hat{c}_k\}$)\\
                \If{$(\mathscr{C},f)$ is feasible and $\MSR(\mathscr{C},f) < U$}{
                    $(\mathscr{C}^*, f^*) \gets (\mathscr{C},f)$\\
                    $U \gets \MSR(\mathscr{C},f)$\\
                }
            }
        }
        \Return $\mathscr{C}^*,f^*$
\end{algorithm}

\maintheorem*
\begin{proof}
    We invoke \Cref{alg:centers-and-radii} for all possible guesses of radius profiles and center assignments. For each of these, we compute the access graph $G$ and invoke \Cref{alg:fair-assignment-general-case} to obtain a solution. 
    By \Cref{lem:final-approximation-ratio}, this solution is feasible and a $(6-\frac{3}{k}+\eps)$-approximation, assuming that \Cref{alg:centers-and-radii} guesses correctly. Since it iterates over all possible guesses, we can be sure that in one of the iterations we do indeed guess correctly. 
    In the end, we return the best solution found, whose cost can therefore be upper bounded by $(6-\frac{3}{k}+\eps)$ times the optimum.
    
    To be able to guess a radius profile, we first need to compute an approximate solution for constrained $k$-center, which can be done in polynomial time. By \Cref{cor:compute-radius-profile}, we can then construct the set $\mathscr{R}$ in Line~\ref{alg-line:radius-guesses} of \Cref{alg:k-MSR-with-constraints} in time $O(k \log_{1+\eps} (k/\eps)^k)$. 
    The outer for-loop in line 3 then goes through $|\mathscr{R}|$ iterations, which is $O(\log_{1+\eps} (k/\eps)^k)$. The inner for-loop in line 4 goes through $k^k$ iterations. 
    So in total, lines 5-10 are invoked $O\left((k \log_{1+\eps}(k/\eps))^k\right)$ times. The runtime of one call to \textsc{Centers-and-radii} is dominated by the runtime of the calls to \textsc{Farthest-first-traversal-completion}. This in turn has the same asymptotic running time as Gonzalez' algorithm, which can be implemented to run in $O(kn)$. So we can bound the runtime of \textsc{Centers-and-radii} by $O(k^2 n)$. The construction of the access graph can be performed in $O(kn)$, as can one call to \textsc{Assignment}. The feasibility of a solution can be checked in $O(n)$. 
    Thus, we obtain an overall running time of $O\left((k \log (k/\eps))^k \cdot \poly(n)\right)$. 
\end{proof} 
\subsubsection{Fairness with two Colors and Equal Proportions} \label{sec:1-1-fairness}
We can get better guarantees for the exact fairness constraint with two colors and equal proportions. 
In this case, we can find a fair assignment such that none of the radii $\hat{r}$ has to be enlarged.
The idea is that we first compute a fair micro clustering (i.e.\ partition $P$ into fair pairs) and then assign these pairs to a common center.

Let $ P = \Gamma_1\cup \Gamma_2$, i.e., $P$ consists of two different colors, $\gamma_1, \gamma_2$. We set $\Gamma_1 = \gamma^{-1}(\gamma_1),\ \Gamma_2 = \gamma^{-1}(\gamma_2)$ as the sets of points carrying the respective color. Here, $\lvert \Gamma_1 \rvert = \lvert \Gamma_2 \rvert = \frac{n}{2}$.
We want to partition the set $P$ into pairs consisting of two points $p_1,p_2$ where $p_1\in \Gamma_1$ and $p_2\in \Gamma_2$. This is equivalent to finding a perfect matching of size $\frac{n}{2}$ between $\Gamma_1$ and $\Gamma_2$.

For this, we construct a flow network where there is an edge between $p_1\in \Gamma_1$ and $p_2\in \Gamma_2$ if and only if they have access to a common center $\hat{c}$, i.e., iff there exists $\hat{c}\in \hat{C}$ with $d(\hat{c},p_1) \leq \hat{r}$ and $d(\hat{c},p_2) \leq \hat{r}$.
We connect all nodes of $\Gamma_1$ to some vertex $s$ and all nodes of $\Gamma_2$ to some vertex $t$.
We set the capacities of all the edges of this network to 1.
Computing a perfect matching between points in $\Gamma_1$ and $\Gamma_2$ corresponds to finding a flow with value $\frac{n}{2}$ in the given network. Such a flow exists if \Cref{alg:centers-and-radii} guessed correctly.

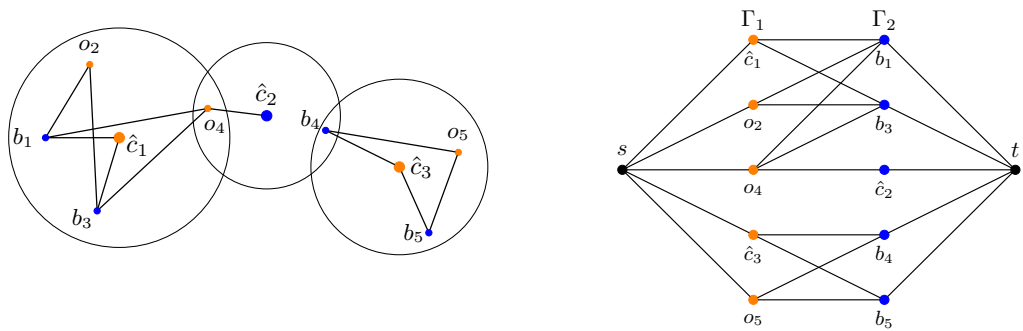
\begin{figure}
    \centering
    \begin{subfigure}[t]{0.5\textwidth}
        \centering
        \resizebox{\linewidth}{!}{
        \begin{tikzpicture}
            \clip (-2,-3) rectangle (5.1,2);
            \def \linecolor  {black}
            \def \centerwidth  {2pt}
            \def \pointwidth  {1.2pt}
            \def \linewidth  {.5pt}
            
            \coordinate (c1) at (0,-0.3);
            \coordinate (c2) at (2,0);
            \coordinate (c3) at (3.8,-0.7);
            \coordinate (c4) at (7,-1.4);
        
            \node at ($(c1) + (0.25,-0.1)$) {$\hat{c}_1$};
            \node at ($(c2) + (0,0.3)$) {$\hat{c}_2$};
            \node at ($(c3) + (0.3,0)$) {$\hat{c}_3$};
            
            \coordinate (12) at ($ (c1) + (1.2,.4) $);
            \coordinate (14) at ($ (c1) + (-.3,-1) $);
            \coordinate (15) at ($ (c1) + (-1,0) $);
            \coordinate (16) at ($ (c1) + (-.4,1) $);
        
            \node at ($(12) + (0.1,-0.25)$) {\small $o_4$};
            \node at ($(14) + (-0.2,-0.1)$) {\small $b_3$};
            \node at ($(15) + (-0.3,0)$) {\small $b_1$};
            \node at ($(16) + (0,0.25)$) {\small $o_2$};
            
            \coordinate (21) at ($ (c2) + (.8,-.2) $);
        
            \node at ($(21) + (-0.2,0.15)$) {\small $b_4$};
            
            \coordinate (31) at ($ (c3) + (.8,.2) $);
            \coordinate (32) at ($ (c3) + (.4,-.9) $);
        
            \node at ($(31) + (0,0.25)$) {\small $o_5$};
            \node at ($(32) + (-0.2,0)$) {\small $b_5$};
            
            
            \foreach \p in {c1,16,12,31,c3}
            {
                \draw[fill=orange, draw=orange] (\p) circle (\pointwidth) {};
            }
            \foreach \p in {c2,14,15,21,32}
            {
                \draw[fill=blue, draw=blue] (\p) circle (\pointwidth) {};
            }
                
            \foreach \p in {c1,c3}
            {
                \draw[fill=orange, draw=orange] (\p) circle (2pt) {};
            }
            \foreach \p in {c2}
            {
                \draw[fill=blue, draw=blue] (\p) circle (2pt) {};
            }
            \draw[draw=black] (c1) circle (1.5) {};
            \draw[draw=black] (c2) circle (1) {};
            \draw[draw=black] (c3) circle (1.2) {};

            \foreach \p in {14,15}
            {
                \draw[-, \linecolor, line width=\linewidth, shorten <=\centerwidth, shorten >=\pointwidth] (c1) -- (\p);
                \draw[-, \linecolor, line width=\linewidth, shorten <=\pointwidth, shorten >=\pointwidth] (12) -- (\p);
                \draw[-, \linecolor, line width=\linewidth, shorten <=\pointwidth, shorten >=\pointwidth] (16) -- (\p);
            }
            
            \draw[-, \linecolor, line width=\linewidth, shorten <=\centerwidth, shorten >=\pointwidth] (c2) -- (12);
            
            \foreach \p in {32,21}
            {
                \draw[-, \linecolor, line width=\linewidth, shorten <=\centerwidth, shorten >=\pointwidth] (c3) -- (\p);
                \draw[-, \linecolor, line width=\linewidth, shorten <=\pointwidth, shorten >=\pointwidth] (31) -- (\p);
            }
        \end{tikzpicture}}
    \end{subfigure}\hfill
    \begin{subfigure}[t]{0.4\textwidth}
        \centering
        \resizebox{\linewidth}{!}{
        \begin{tikzpicture}
            \def \linewidth {.5pt}
            \def \pointwidth {2pt}
            \def \linecolor {black}
        
            \def \xA {2}
            \def \xB {4}
            \coordinate (s) at (0,0);
            \coordinate (t) at (6,0);
            
            \coordinate (A1) at (\xA,2);
            \coordinate (A2) at (\xA,1);
            \coordinate (A3) at (\xA,0);
            \coordinate (A4) at (\xA,-1);
            \coordinate (A5) at (\xA,-2);
        
            \node at ($(A1) + (0,-0.3)$) {\small $\hat{c}_1$};
            \node at ($(A2) + (0,-0.3)$) {\small $o_2$};
            \node at ($(A3) + (0,-0.3)$) {\small $o_4$};
            \node at ($(A4) + (0,-0.3)$) {\small $\hat{c}_3$};
            \node at ($(A5) + (0,-0.3)$) {\small $o_5$};
            
            \coordinate (B1) at (\xB,2);
            \coordinate (B2) at (\xB,1);
            \coordinate (B3) at (\xB,0);
            \coordinate (B4) at (\xB,-1);
            \coordinate (B5) at (\xB,-2);
        
            \node at ($(B1) + (0,-0.3)$) {\small $b_1$};
            \node at ($(B2) + (0,-0.3)$) {\small $b_3$};
            \node at ($(B3) + (0,-0.3)$) {\small $\hat{c}_2$};
            \node at ($(B4) + (0,-0.3)$) {\small $b_4$};
            \node at ($(B5) + (0,-0.3)$) {\small $b_5$};
            
            \foreach \p in {s,t}
            {
                \draw[fill=black, draw=black] (\p) circle (2pt) {};
            }
            \node at ($(s) + (0,0.3)$) {$s$};
            \node at ($(t) + (0,0.3)$) {$t$};
            \foreach \p in {A1,A2,A3,A4,A5}
            {
                \draw[fill=orange, draw=orange] (\p) circle (2pt) {};
                    \draw[-, \linecolor, line width=\linewidth, shorten <=\pointwidth, shorten >=\pointwidth] (s) -- (\p);
            }
            \node at ($(A1) + (0,0.3)$) {$\Gamma_1$};
            \foreach \p in {B1,B2,B3,B4,B5}
            {
                \draw[fill=blue, draw=blue] (\p) circle (2pt) {};
                    \draw[-, \linecolor, line width=\linewidth, shorten <=\pointwidth, shorten >=\pointwidth] (\p) -- (t);
            }
            \node at ($(B1) + (0,0.3)$) {$\Gamma_2$};
            
            \foreach \p in {A1,A2,A3}{
                \draw[-, \linecolor, line width=\linewidth, shorten <=\pointwidth, shorten >=\pointwidth] (\p) -- (B1);
                \draw[-, \linecolor, line width=\linewidth, shorten <=\pointwidth, shorten >=\pointwidth] (\p) -- (B2);
            }    
            \draw[-, \linecolor, line width=\linewidth, shorten <=\pointwidth, shorten >=\pointwidth] (A3) -- (B3);
            \foreach \p in {A4,A5}{
                \draw[-, \linecolor, line width=\linewidth, shorten <=\pointwidth, shorten >=\pointwidth] (\p) -- (B4);
                \draw[-, \linecolor, line width=\linewidth, shorten <=\pointwidth, shorten >=\pointwidth] (\p) -- (B5);
            } 
        \end{tikzpicture}}
    \end{subfigure}
    \caption{A fair $k$-min-sum-radii instance with two colors and equal proportions. Left: Output of \Cref{alg:centers-and-radii} as balls $B(\hat{c}_i,\hat{r}_i)$ for $i=1,2,3$ and the graph edges between any fair pair of points that have access to the same center $\hat{c}$. Right: The corresponding flow network. All edges have capacity 1.}
    \label{fig:enter-label}
\end{figure}

From the flow, we can construct the fair pairs by combining two points $p_1\in \Gamma_1$ and $p_2\in \Gamma_2$ if the edge connecting them carries flow. We assign such a pair to a center $\hat{c}$ to which both points have access. We can summarize this in the following observation:
\begin{observation}\label{obs:1-fairlets}
    Let $F = \{ p_1, p_2\}$ be a fair pair constructed as described above, let $\hat{c}$ be the center to which it gets assigned and $\hat{r}$ the corresponding radius. Then
    $\max_{i=1,2}d(p_i,\hat{c})\leq \hat{r}$.
\end{observation}
In other words, assigning the fair pairs as a whole does not increase the cost.
Together with \Cref{lem:algo-apx}, this implies the following theorem.
\begin{theorem}\label{thm-1-1}
    Let $P = \Gamma_1 \cup \Gamma_2$ be a set of points consisting of only two different color groups $\Gamma_1$ and $\Gamma_2$ that fulfill $\lvert\Gamma_1\rvert = \lvert \Gamma_2\rvert$.
    For every $\eps > 0$, there exists an FPT $(3+\eps)$-approximation algorithm for the fair $k$-min-sum-radii problem.
\end{theorem}
\begin{proof}
    We run Algorithm~\ref{alg:centers-and-radii} to obtain a set of balls $B(\hat{c},\hat{r})$ and then compute a partitioning of $P$ into fair pairs as described above. 
    We assign each fair pair to a center to which both points of the pair have access. 
    By \Cref{lem:algo-apx}, $ \sum_{i=1}^k \hat{r}_i \leq 3(1+\eps)\sum_{i=1}^k r^*_i $. As noted in \Cref{obs:1-fairlets}, assigning fair pairs does not increase the cost, which concludes the proof.
\end{proof}

\subsubsection{Uniform Lower Bounds} \label{sec:lower-bounds}
In $k$-min-sum-radii with uniform lower bounds, we have an additional input number $\ell \in \mathbb{N}$, and every cluster in the solution needs to contain at least $\ell$ points. In this case we set up a network flow between the centers $\hat{C}$ on the left and the points on the right. There is a super source $s$ and an edge $(s,\hat{c})$ for every $\hat{c} \in \hat{C}$. The capacity of these edges is set to the lower bound $L$. Then every $\hat{c}$ is connected to all $x \in P$ with $d(\hat{c},x) \le \hat{r}$. Finally, all points are connected with a unit capacity edge $(x,t)$ to a super sink $t$. Any flow in this network corresponds to an assignment of at least $L$ points to each center. 
Since our balls cover the optimum solution, we know that there exists a feasible flow in this network that sends $k \cdot L $ units of flow to the super sink. We can thus run a maximum flow algorithm to find such an assignment. After that, any remaining point $x \in P$ can be assigned arbitrarily to a center $\hat{c}$ with $d(x,\hat{c})\le \hat{r}$. Again, this is possible due to \Cref{lem:injective-mapping}.

\begin{theorem}\label{cor-lower}
    There exists an FPT $(3+\eps)$-approximation algorithm for the $k$-min-sum-radii problem with uniform lower bounds.
\end{theorem}

\bibliography{main}

\newpage
\appendix

\section{Guessing the Radius Profile}\label{sec:guessing-radii}

\newcommand{\Center}{C}
\newcommand{\radiusCenter}{r^{\Center, \alpha}}
\newcommand{\radiusCenteropt}{r^{\Center, *}}
\newcommand{\radiusMSR}{r^{\MSR, \beta}_{\max}}
\newcommand{\radiusMSRopt}{r^{\MSR, *}}
Let $S^* = (\mathscr{C}^*, \sigma^*)$ be an optimal $k$-min-sum-radii solution with radius profile $(r^*_1, \ldots, r^*_k)$. We now show how to derive a radius profile $(\tilde{r}_1, \ldots, \tilde{r}_k)$ from an approximate $k$-center solution, with the property that $r^*_i \leq \tilde{r}_i \leq (1+\eps)r^*_i$ for $1 \leq i \leq k$. We start with the following observation.
\begin{lemma} \label{lem:largest-radius}
    Let $F_0$ be the value of a $\beta$-approximate solution for $k$-center with mergeable constraints and $r^*_1$ the largest radius of an optimal $k$-min-sum-radii solution for the same instance. Then it holds that
    $r^*_1 \in \left[\frac{F_0}{\beta}, k F_0 \right]$.
\end{lemma}

\begin{proof}
    First note that any feasible solution for $k$-center with mergeable constraints is also a feasible solution for $k$-min-sum-radii with mergeable constraints. Since $F_0$ is the largest radius in the $k$-center solution, the sum of the radii in that solution is at most $k F_0$. Then we must have $r_1^* \leq kF_0$, as otherwise the $k$-center solution would yield a better $k$-min-sum-radii solution than the optimal $k$-min-sum-radii-solution itself.
    For the lower bound, consider an \emph{optimal} $k$-center solution value $F^*$. Since $F_0$ is an $\beta$-approximation, we have $F_0 \leq \beta F^*$. Additionally, we must have $r^*_1 \geq F^*$, as otherwise the $k$-min-sum-radii solution would be a better solution to $k$-center. Combining these two insights we get $r^*_1 \geq F^*\geq F_0/\beta$.
\end{proof}
Now that we have an interval that is certain to contain the value of the largest radius, we can construct a set of size $O(\log k)$ that contains a $(1+\eps)$-approximation to that value.
\begin{lemma} \label{lem:largest-radius-guess}
    Let $F_0$ be the value of a $\beta$-approximate solution for $k$-center with mergeable constraints and $r_1^*$ the largest radius of an optimal $k$-min-sum-radii solution for the same instance. Let $\eps > 0$. Then the set $R = \{(1+\eps)^j \frac{F_0}{\beta} \mid j \leq \ceil{\log_{1+\eps}(\beta k)}\}$ contains a value $\tilde{r}_1$ such that $r_1^* \leq \tilde{r}_1 \leq (1+\eps)r_1^*$.
\end{lemma}
\begin{proof}
    By \Cref{lem:largest-radius} we have $r_1^* \in I = \left[ \frac{F_0}{\beta}, kF_0 \right]$. We cover this interval with smaller intervals $I_j = \left[(1+\eps)^{j-1} \frac{F_0}{\beta}, (1+\eps)^{j} \frac{F_0}{\beta}\right]$ for $j\in \{1, \ldots, \ceil{\log_{1+\eps}(\beta k)})\}$. Note that the union of the $I_j$ contains $I$. This means that any number $r^*_k\in I$ must lie in one of these intervals, say $I_\ell = \left[(1+\eps)^{\ell-1} \frac{F_0}{\beta}, (1+\eps)^{\ell} \frac{F_0}{\beta} \right]$. Choosing the right endpoint $\tilde{r}_1 = (1+\eps)^\ell \frac{F_0}{\beta}$ gives us a number with the desired property. Since the set $R$ consists exactly of the endpoints of all the $I_j$, the claim follows.
 
\end{proof}
Once a guess for the largest radius is fixed, we can argue how to (approximately) guess the remaining radii. To this end, we assume that all radii are at least $(\eps/k)r_1^*$. However, this is not a very limiting assumption, as rounding up all smaller radii to that value can increase the cost of any solution $S$ by at most $\eps r_1^* \leq \eps \msr(S)$.\\
The following lemma implies that only $O(\log\frac{k}{\eps})$ guesses must be made to find a good approximation for any smaller radius of the $k$-min-sum-radii solution.

\begin{lemma} \label{lem:guess-radii}
    Let $r_1^*$ be the largest radius of an optimal $k$-min-sum-radii solution.
    If $r_j^* \geq \frac{\eps}{k}r_1^*$ for all $1\leq j \leq k-1$, then the set $R' = \{(1+\eps)^j\frac{\eps}{k}r_1^* \mid j\leq \ceil{\log_{1+\eps}(\frac{k}{\eps})}\}$ contains values $\tilde{r}_{2}, \ldots, \tilde{r}_{k}$ such that $r_j^* \leq \tilde{r}_j \leq (1+\eps)r_j^*$ for all $2\leq j \leq k$.
\end{lemma}
\begin{proof}
    We employ the same technique as in \Cref{lem:largest-radius-guess}. Since all remaining radii $\tilde{r}_{2}, \ldots, \tilde{r}_k$ must fall into the interval $[(\eps/k)r^*_1, r^*_1]$, we can again cover this by smaller intervals. The set $R'$ is the set of right endpoints of these intervals and by the same arguments as in the proof for \Cref{lem:largest-radius-guess}, for each $r_j^*$ we get a number $\tilde{r}_j\in R'$ with $r_j^* \leq \tilde{r}_j \leq (1+\eps)r_j^*$.
\end{proof}
Combining the above statements, we can prove the corollary from the main body of the paper.

\computeradiusprofile*
\begin{proof}
    By \Cref{lem:largest-radius-guess}, there are $O(\log_{1+\eps} k)$ candidate values for $\tilde{r}_1$ such that one of them is a $(1+\eps)$-approximation for $r^*_1$. By \Cref{lem:guess-radii}, for each of these candidates, there are $O(\log_{1+\eps} (k/\eps))$ candidates for each of the remaining $k-1$ radii. So in total we have $O\left(\log_{1+\eps}k \cdot \log_{1+\eps}(k/\eps)^{k-1}\right) = O\left(\log_{1+\eps}(k/\eps)^k\right)$ total radius profiles, and one of these is sure to be near-optimal. Since each radius profile consists of $k$ numbers, pre-computing all possible profiles takes at most $O\left(k \log_{1+\eps} (k/\eps)^k\right)$
\end{proof}

\section{\texorpdfstring{$k$-Center with Mergeable Constraints}{k-Center with Mergeable Constraints}}\label{fairness-notions}
For guessing the radius profiles for the constrained $k$-min-sum-radii problem as described in \Cref{sec:guessing-radii}, we need constant-factor approximations for the respective constrained $k$-center problem. In the following, we list some mergeable constraints with the approximation factors that are known for them.

\paragraph*{Exact Fairness}
Here, we refer to the fairness definition already stated in \Cref{sec:prelims}.
R\"osner and Schmidt \cite{rosner2018privacy} give a 12-approximation for the case of multiple colors and no restriction on initial ratios.
Bercea et al.\ improve it to a 5-approximation for the most general case\cite{bercea2019cost}.
For 1:1 case Chierichetti et al. give a $3$-approximation \cite{CKLV17} (if the matching in that algorithm is set up correctly). 

Exact fairness is a mergeable constraint: Let $(\mathscr{C},\sigma)$ be a clustering that is exactly fair. Consider a pair of clusters $C_i \coloneqq \sigma^{-1}(c_i)$, $C_j \coloneqq \sigma^{-1}(c_j)$. As $C_i\cap C_j = \emptyset$,
\[ \frac{|(C_i\cup C_j) \cap \Gamma_{\ell}|}{|C_i\cup C_j|} = \frac{|C_i \cap \Gamma_{\ell}| + |C_j \cap \Gamma_{\ell}|}{|C_i|+|C_j|} = \frac{|\Gamma_{\ell}|}{|P|},\]
where the last equality is implied by the following observation:
If $\frac{a}{b} = \frac{c}{d}$ for some $a,b,c,d\in \mathbb{R}^+$, then adding up enumerators and denominators maintains this ratio: $\frac{a+c}{b+d} = \frac{a}{b}$.  

\paragraph*{Ratio Balance}
Let $b \in [0,1]$. A clustering with only two colors $\Gamma_1, \Gamma_2$ is said to be $b$-\textit{balanced} if \[ \min \left\{\frac{\lvert \Gamma_1\cap C \rvert}{\lvert \Gamma_2\cap C \rvert}, \frac{\lvert \Gamma_2\cap C \rvert}{\lvert \Gamma_1\cap C \rvert}\right\} \geq b \]
for all clusters $C$.
This is the model initially studied by 
Chierichetti et al. They give a 4-approximation for two colors for $\frac{1}{t}$-balance for $t\in \mathbb{N}$ and a 3-approximation if $t=1$ \cite{CKLV17}.

When merging two clusters, the resulting balance is at least as high as the balance of any of the individual clusters. Therefore, ratio balance is a mergeable constraint.

\paragraph*{Exact Balance}
Exact balance generalizes the balance notion to more than two colors. However, it is stricter in the sense that it requires the number of points of different colors to be the same in every cluster, i.e.\
\[ \lvert C \cap \Gamma_1 \rvert = \lvert C \cap \Gamma_2 \rvert = \ldots = \lvert C \cap \Gamma_m \rvert \]
for all clusters $C$ and colors $\Gamma_1, \ldots, \Gamma_m$.
B\"ohm et al.\ give a 4-approximation for this problem \cite{bohm2020fair}.

It is easy to see that exact balance is a mergeable constraint since if $a_1=b_1$ and $a_2=b_2$, then $a_1 + a_2 = b_1 + b_2$.

\paragraph*{\texorpdfstring{$(\ell,u)$-fairness}{(l,u)-fairness}}
Bercea et al.\ \cite{bercea2019cost} give a different fairness notion that works for $m\geq 2$ colors $\Gamma_1,\ldots, \Gamma_m$. 
Let $\ell = (\ell_1, \ell_2, \ldots, \ell_m)$ and $u = (u_1, u_2, \ldots, u_m)$ for $\ell_1,\ldots,\ell_m,u_1,\ldots,u_m \in \mathbb{Q}$. 
A clustering is said to be $(\ell,u)$-fair if 
\[ \ell_i \leq \frac{\lvert C \cap \Gamma_i \rvert}{\lvert C \rvert} \leq u_i \]
for all clusters $C$ and all $i\leq m$.
The exact balance notion described earlier is the special case in which $\ell_i = u_i = \frac{\lvert P \cap \Gamma_i \rvert}{\lvert P \rvert}$ for all $i\leq m$.
Bercea et al.\ \cite{bercea2019cost} give a bicriteria approximation for this case which is $3$-approximate but violates the fairness constraints by at most $1$ point per cluster. Similarly, Bera et al.\ show how to obtain a 4-approximation if fairness constraints can be violated additively \cite{bera2019fair}.
To see that $(\ell,u)$-fairness is a mergeable constraint, let $C_1,C_2$ be two disjoint clusters that fulfill the constraint.
That is, $l_i\lvert C_1\rvert \leq \lvert C_1\cap \Gamma_i \rvert \leq u_i\lvert C_1\rvert$ and $l_i\lvert C_2\rvert \leq \lvert C_2\cap \Gamma_i \rvert \leq u_i\lvert C_2\rvert$.
Then,
\[ l_i \lvert C_1 \cup C_2 \rvert = l_i\lvert C_1 \rvert + l_i\lvert C_2 \rvert \leq \lvert C_1\cap \Gamma_i \rvert + \lvert C_2\cap \Gamma_i \rvert = \lvert (C_1\cup C_2)\cap \Gamma_i \rvert \]
and
\[ 
\lvert (C_1\cup C_2)\cap \Gamma_i \rvert = \lvert C_1\cap \Gamma_i \rvert + \lvert C_2\cap \Gamma_i \rvert \leq u_i\lvert C_1 \rvert + u_i\lvert C_2 \rvert = u_i\lvert C_1 \cup C_2 \rvert
\]
for all $i\leq m$.

\section{\texorpdfstring{An Example Run of \Cref{alg:centers-and-radii}}{An example run of Algorithm 1}} \label{sec:example-run}
\Cref{fig:example-run} shows an example run of \Cref{alg:centers-and-radii} specified in \Cref{sec:initial-selection-of-centers-and-radii}.
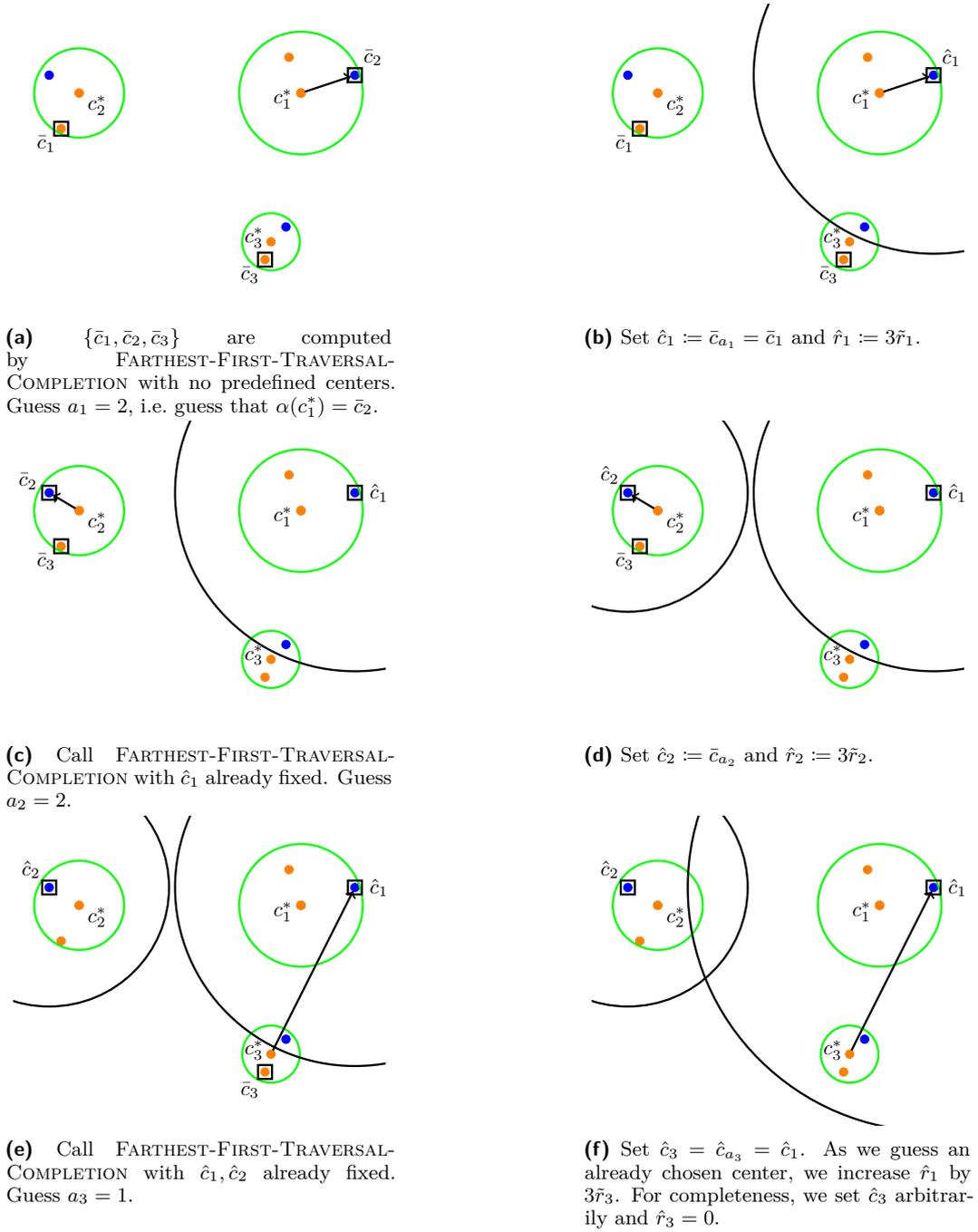
\begin{figure}[t]
\centering
\begin{subfigure}[t]{0.40\textwidth}
\centering
\resizebox{\linewidth}{!}{
\begin{tikzpicture}
    \clip (-1.8,-3.7) rectangle (4.4,1.5);
    \def \linecolor  {black}
    \def \pointwidth  {2pt}
    \def \centerwidth  {2pt}
    \def \linewidth  {1pt}
    \def \step {1}
    
    \coordinate (1) at (3,0);
    \coordinate (11) at ($(1) + (-.2,.6)$);
    \coordinate (12) at ($(1) + (.9,.3)$);

    \coordinate (2) at (-.7,0);
    \coordinate (21) at ($(2) + (-.3,-.6)$);
    \coordinate (22) at ($(2) + (-.5,.3)$);
    
    \coordinate (3) at (2.5,-2.5);
    \coordinate (31) at ($(3) + (-.1,-.3)$);
    \coordinate (32) at ($(3) + (.25, .25)$);
    
    \foreach \p in {2,21,1,11,3,31}
    {
        \draw[fill=orange, draw=orange] (\p) circle (2pt) {};
    }
    \foreach \p in {22,12,32}
    {
        \draw[fill=blue, draw=blue] (\p) circle (2pt) {};
    }

    \draw[draw=green, line width=\linewidth] (2) circle (.75) {};
    \draw[fill=orange, draw=orange] (2) circle (2pt) {};
    \draw[draw=green, line width=\linewidth] (1) circle (1.03) {};

    \node at ($(2) + (0.3,-0.2)$) {$c_2^*$};
    \node at ($(1) + (-.3,-.1)$) {$c_1^*$};
    \node at ($(3) + (-0.28,0.07)$) {$c_3^*$};
    
    \draw[fill=orange, draw=orange] (1) circle (2pt) {};
    \draw[draw=green, line width=\linewidth] (3) circle (.48) {};
    \draw[fill=orange, draw=orange] (1) circle (2pt) {};

    \foreach \p in {12,21,31}
    {
        \draw[draw=black, line width=1pt] (\p) node[rectangle, minimum height=.2cm,minimum width=.2cm,draw] (3.5pt) {};
    }
    
    \node at ($(21) + (-0.25,-0.25)$) {$\bar{c}_1$};
    \node at ($(12) + (.3,.3)$) {$\bar{c}_2$};
    \node at ($(31) + (-0.25,-0.25)$) {$\bar{c}_3$};
    
    \draw[->, \linecolor, line width=\linewidth, shorten <=\pointwidth, shorten >=\centerwidth] (1) -- (12);
\end{tikzpicture}
}
\caption{$\{\bar{c}_1,\bar{c}_2,\bar{c}_3\}$ are computed by \textsc{Farthest-First-Traversal-Completion} with no predefined centers. Guess $a_1 = 2$, i.e.\ guess that $\alpha(c_1^*)=\bar{c}_2$.}
\end{subfigure}\hfill
\begin{subfigure}[t]{0.40\textwidth}
\centering
\resizebox{\linewidth}{!}{
\begin{tikzpicture}
    \clip (-1.8,-3.7) rectangle (4.4,1.5);
    \def \linecolor  {black}
    \def \pointwidth  {2pt}
    \def \centerwidth  {2pt}
    \def \linewidth  {1pt}
    \def \step {1}
    
    \coordinate (1) at (3,0);
    \coordinate (11) at ($(1) + (-.2,.6)$);
    \coordinate (12) at ($(1) + (.9,.3)$);

    \coordinate (2) at (-.7,0);
    \coordinate (21) at ($(2) + (-.3,-.6)$);
    \coordinate (22) at ($(2) + (-.5,.3)$);
    
    \coordinate (3) at (2.5,-2.5);
    \coordinate (31) at ($(3) + (-.1,-.3)$);
    \coordinate (32) at ($(3) + (.25, .25)$);
    
    \foreach \p in {2,21,1,11,3,31}
    {
        \draw[fill=orange, draw=orange] (\p) circle (2pt) {};
    }
    \foreach \p in {22,12,32}
    {
        \draw[fill=blue, draw=blue] (\p) circle (2pt) {};
    }
    \node at ($(12) + (.3,.3)$) {$\hat{c}_1$};
    \draw[draw=green, line width=\linewidth] (2) circle (.75) {};
    \draw[fill=orange, draw=orange] (2) circle (2pt) {};
    \draw[draw=green, line width=\linewidth] (1) circle (1.03) {};

    \node at ($(2) + (0.3,-0.2)$) {$c_2^*$};
    \node at ($(1) + (-.3,-.1)$) {$c_1^*$};
    \node at ($(3) + (-0.28,0.07)$) {$c_3^*$};
    
    \draw[fill=orange, draw=orange] (1) circle (2pt) {};
    \draw[draw=green, line width=\linewidth] (3) circle (.48) {};
    \draw[fill=orange, draw=orange] (1) circle (2pt) {};

    \foreach \p in {12,21,31}
    {
        \draw[draw=black, line width=1pt] (\p) node[rectangle, minimum height=.2cm,minimum width=.2cm,draw] (3.5pt) {};
    }
    \node at ($(21) + (-0.25,-0.25)$) {$\bar{c}_1$};
    \node at ($(31) + (-0.25,-0.25)$) {$\bar{c}_3$};
    
    \draw[->, \linecolor, line width=\linewidth, shorten <=\pointwidth, shorten >=\centerwidth] (1) -- (12);
    \draw[draw=black, line width=\linewidth] (12) circle (3) {};
\end{tikzpicture}
}
\caption{Set $\hat{c}_1 \coloneqq \bar{c}_{a_1} = \bar{c}_1$ and $\hat{r}_1 \coloneqq 3\tilde{r}_1$.}
\end{subfigure}\hfill\\
\begin{subfigure}[t]{0.40\textwidth}
\centering
\resizebox{\linewidth}{!}{
\begin{tikzpicture}
    \clip (-1.8,-3.7) rectangle (4.4,1.5);
    \def \linecolor  {black}
    \def \pointwidth  {2pt}
    \def \centerwidth  {2pt}
    \def \linewidth  {1pt}
    \def \step {1}
    
    \coordinate (1) at (3,0);
    \coordinate (11) at ($(1) + (-.2,.6)$);
    \coordinate (12) at ($(1) + (.9,.3)$);

    \coordinate (2) at (-.7,0);
    \coordinate (21) at ($(2) + (-.3,-.6)$);
    \coordinate (22) at ($(2) + (-.5,.3)$);
    
    \coordinate (3) at (2.5,-2.5);
    \coordinate (31) at ($(3) + (-.1,-.3)$);
    \coordinate (32) at ($(3) + (.25, .25)$);
    
    \foreach \p in {2,21,1,11,3,31}
    {
        \draw[fill=orange, draw=orange] (\p) circle (2pt) {};
    }
    \foreach \p in {22,12,32}
    {
        \draw[fill=blue, draw=blue] (\p) circle (2pt) {};
    }

    \draw[draw=green, line width=\linewidth] (2) circle (.75) {};
    \draw[fill=orange, draw=orange] (2) circle (2pt) {};
    \draw[draw=green, line width=\linewidth] (1) circle (1.03) {};

    \node at ($(2) + (0.3,-0.2)$) {$c_2^*$};
    \node at ($(1) + (-.3,-.1)$) {$c_1^*$};
    \node at ($(3) + (-0.28,0.07)$) {$c_3^*$};
    
    \draw[fill=orange, draw=orange] (1) circle (2pt) {};
    \draw[draw=green, line width=\linewidth] (3) circle (.48) {};
    \draw[fill=orange, draw=orange] (1) circle (2pt) {};

    \foreach \p in {22,21,12}
    {
        \draw[draw=black, line width=1pt] (\p) node[rectangle, minimum height=.2cm,minimum width=.2cm,draw] (3.5pt) {};
    }
    \node at ($(22) + (-0.35,0.2)$) {$\bar{c}_2$};
    \node at ($(21) + (-0.25,-0.25)$) {$\bar{c}_3$};
    
    \node at ($(12) + (0.4,0)$) {$\hat{c}_1$};
    \draw[draw=black, line width=\linewidth] (12) circle (3) {};
    
    \draw[->, \linecolor, line width=\linewidth, shorten <=\pointwidth, shorten >=\centerwidth] (2) -- (22);
\end{tikzpicture}
}
\caption{Call \textsc{Farthest-First-Traversal-Completion} with $\hat{c}_1$ already fixed. Guess $a_2=2$.}
\end{subfigure}\hfill
\begin{subfigure}[t]{0.40\textwidth}
\centering
\resizebox{\linewidth}{!}{
\begin{tikzpicture}
    \clip (-1.8,-3.7) rectangle (4.4,1.5);
    \def \linecolor  {black}
    \def \pointwidth  {2pt}
    \def \centerwidth  {2pt}
    \def \linewidth  {1pt}
    \def \step {1}
    
    \coordinate (1) at (3,0);
    \coordinate (11) at ($(1) + (-.2,.6)$);
    \coordinate (12) at ($(1) + (.9,.3)$);

    \coordinate (2) at (-.7,0);
    \coordinate (21) at ($(2) + (-.3,-.6)$);
    \coordinate (22) at ($(2) + (-.5,.3)$);
    
    \coordinate (3) at (2.5,-2.5);
    \coordinate (31) at ($(3) + (-.1,-.3)$);
    \coordinate (32) at ($(3) + (.25, .25)$);
    
    \foreach \p in {2,21,1,11,3,31}
    {
        \draw[fill=orange, draw=orange] (\p) circle (2pt) {};
    }
    \foreach \p in {22,12,32}
    {
        \draw[fill=blue, draw=blue] (\p) circle (2pt) {};
    }

    \draw[draw=green, line width=\linewidth] (2) circle (.75) {};
    \draw[fill=orange, draw=orange] (2) circle (2pt) {};
    \draw[draw=green, line width=\linewidth] (1) circle (1.03) {};

    \node at ($(2) + (0.3,-0.2)$) {$c_2^*$};
    \node at ($(1) + (-.3,-.1)$) {$c_1^*$};
    \node at ($(3) + (-0.28,0.07)$) {$c_3^*$};
    
    \draw[fill=orange, draw=orange] (1) circle (2pt) {};
    \draw[draw=green, line width=\linewidth] (3) circle (.48) {};
    \draw[fill=orange, draw=orange] (1) circle (2pt) {};

    \foreach \p in {22,21,12}
    {
        \draw[draw=black, line width=1pt] (\p) node[rectangle, minimum height=.2cm,minimum width=.2cm,draw] (3.5pt) {};
    }
    \node at ($(21) + (-0.25,-0.25)$) {$\bar{c}_3$};
    
    \node at ($(12) + (0.4,0)$) {$\hat{c}_1$};
    \draw[draw=black, line width=\linewidth] (12) circle (3) {};
    
    \draw[->, \linecolor, line width=\linewidth, shorten <=\pointwidth, shorten >=\centerwidth] (2) -- (22);
    \node at ($(22) + (-.3,.3)$) {$\hat{c}_2$};
    \draw[draw=black, line width=\linewidth] (22) circle (2) {};
\end{tikzpicture}
}
\caption{Set $\hat{c}_2\coloneqq\bar{c}_{a_2}$ and $\hat{r}_2 \coloneqq 3\tilde{r}_2$.}
\end{subfigure}\hfill\\
\begin{subfigure}[t]{0.40\textwidth}
\centering
\resizebox{\linewidth}{!}{
\begin{tikzpicture}
    \clip (-1.8,-3.7) rectangle (4.4,1.5);
    \def \linecolor  {black}
    \def \pointwidth  {2pt}
    \def \centerwidth  {2pt}
    \def \linewidth  {1pt}
    \def \step {1}
    
    \coordinate (1) at (3,0);
    \coordinate (11) at ($(1) + (-.2,.6)$);
    \coordinate (12) at ($(1) + (.9,.3)$);

    \coordinate (2) at (-.7,0);
    \coordinate (21) at ($(2) + (-.3,-.6)$);
    \coordinate (22) at ($(2) + (-.5,.3)$);
    
    \coordinate (3) at (2.5,-2.5);
    \coordinate (31) at ($(3) + (-.1,-.3)$);
    \coordinate (32) at ($(3) + (.25, .25)$);
    
    \foreach \p in {2,21,1,11,3,31}
    {
        \draw[fill=orange, draw=orange] (\p) circle (2pt) {};
    }
    \foreach \p in {22,12,32}
    {
        \draw[fill=blue, draw=blue] (\p) circle (2pt) {};
    }

    \draw[draw=green, line width=\linewidth] (2) circle (.75) {};
    \draw[fill=orange, draw=orange] (2) circle (2pt) {};
    \draw[draw=green, line width=\linewidth] (1) circle (1.03) {};

    \node at ($(2) + (0.3,-0.2)$) {$c_2^*$};
    \node at ($(1) + (-.3,-.1)$) {$c_1^*$};
    \node at ($(3) + (-0.28,0.07)$) {$c_3^*$};
    
    \draw[fill=orange, draw=orange] (1) circle (2pt) {};
    \draw[draw=green, line width=\linewidth] (3) circle (.48) {};
    \draw[fill=orange, draw=orange] (1) circle (2pt) {};

    \foreach \p in {22,12,31}
    {
        \draw[draw=black, line width=1pt] (\p) node[rectangle, minimum height=.2cm,minimum width=.2cm,draw] (3.5pt) {};
    }
    \node at ($(31) + (-0.25,-0.3)$) {$\bar{c}_3$};
    
    \node at ($(12) + (0.4,0)$) {$\hat{c}_1$};
    \draw[draw=black, line width=\linewidth] (12) circle (3) {};
    
    \draw[->, \linecolor, line width=\linewidth, shorten <=\pointwidth, shorten >=\centerwidth] (3) -- (12);
    \node at ($(22) + (-.3,.3)$) {$\hat{c}_2$};
    \draw[draw=black, line width=\linewidth] (22) circle (2) {};
\end{tikzpicture}
}
\caption{Call \textsc{Farthest-First-Traversal-Completion} with $\hat{c}_1, \hat{c}_2$ already fixed. Guess $a_3=1$.}
\end{subfigure}\hfill
\begin{subfigure}[t]{0.40\textwidth}
\centering
\resizebox{\linewidth}{!}{
\begin{tikzpicture}
    \clip (-1.8,-3.7) rectangle (4.4,1.5);
    \def \linecolor  {black}
    \def \pointwidth  {2pt}
    \def \centerwidth  {2pt}
    \def \linewidth  {1pt}
    \def \step {1}
    
    \coordinate (1) at (3,0);
    \coordinate (11) at ($(1) + (-.2,.6)$);
    \coordinate (12) at ($(1) + (.9,.3)$);

    \coordinate (2) at (-.7,0);
    \coordinate (21) at ($(2) + (-.3,-.6)$);
    \coordinate (22) at ($(2) + (-.5,.3)$);
    
    \coordinate (3) at (2.5,-2.5);
    \coordinate (31) at ($(3) + (-.1,-.3)$);
    \coordinate (32) at ($(3) + (.25, .25)$);
    
    \foreach \p in {2,21,1,11,3,31}
    {
        \draw[fill=orange, draw=orange] (\p) circle (2pt) {};
    }
    \foreach \p in {22,12,32}
    {
        \draw[fill=blue, draw=blue] (\p) circle (2pt) {};
    }

    \draw[draw=green, line width=\linewidth] (2) circle (.75) {};
    \draw[fill=orange, draw=orange] (2) circle (2pt) {};
    \draw[draw=green, line width=\linewidth] (1) circle (1.03) {};

    \node at ($(2) + (0.3,-0.2)$) {$c_2^*$};
    \node at ($(1) + (-.3,-.1)$) {$c_1^*$};
    \node at ($(3) + (-0.28,0.07)$) {$c_3^*$};
    
    \draw[fill=orange, draw=orange] (1) circle (2pt) {};
    \draw[draw=green, line width=\linewidth] (3) circle (.48) {};
    \draw[fill=orange, draw=orange] (1) circle (2pt) {};

    \foreach \p in {22,12}
    {
        \draw[draw=black, line width=1pt] (\p) node[rectangle, minimum height=.2cm,minimum width=.2cm,draw] (3.5pt) {};
    }
    \node at ($(12) + (0.4,0)$) {$\hat{c}_1$};
    \draw[draw=black, line width=\linewidth] (12) circle (4.1) {};
    
    \draw[->, \linecolor, line width=\linewidth, shorten <=\pointwidth, shorten >=\centerwidth] (3) -- (12);
    \node at ($(22) + (-.3,.3)$) {$\hat{c}_2$};
    \draw[draw=black, line width=\linewidth] (22) circle (2) {};
\end{tikzpicture}
}
\caption{Set $\hat{c}_3 = \hat{c}_{a_3} = \hat{c}_1$. As we guess an already chosen center, we increase $\hat{r}_1$ by $3\tilde{r}_3$. For completeness, we set $\hat{c}_3$ arbitrarily and $\hat{r}_3=0$.}
\end{subfigure}\hfill
\caption{An example run of \Cref{alg:centers-and-radii}. Given a set of points of which one-third are blue and two-thirds are orange. The green circles indicate the optimal clusters with centers $c_1^*,c_2^*,c_3^*$. Black squares indicate the centers $\{\hat{c}_1,\ldots,\hat{c}_i,\bar{c}_{i+1},\ldots,\bar{c}_k\}$ of the $k$-center completion solution output by \Cref{alg:gonzalez} on given centers $\{\hat{c}_1,\ldots,\hat{c}_i\}$. The arrows represent the guess of the assignment of the next optimal center. Here, we assume that the algorithm guesses correctly.}
\label{fig:example-run}
\end{figure}

\end{document}